\RequirePackage{amsmath}
\documentclass[letterpaper,11pt,english]{article}
\usepackage[letterpaper, portrait, margin=2.5cm]{geometry}

\usepackage{amssymb, amsfonts}% mathematical symbols and the like
\usepackage{graphicx}
\usepackage{amsthm}
\usepackage{algorithm}% algorithm floats
\usepackage[noend]{algpseudocode}% algorithm typesetting
\usepackage[nolist]{acronym}% consistently use acronyms
\usepackage[capitalize,nameinlink,noabbrev]{cleveref}
\usepackage{todonotes}
\usepackage{mathtools}
\usepackage{thm-restate}
\newacro{ses}[SES]{smallest enclosing sphere}
\usepackage{csquotes}
\usepackage[inline,shortlabels]{enumitem}

\usepackage{todonotes}
\usepackage[compatibility=false]{caption}
\usepackage{subcaption}

\newtheorem{theorem}{Theorem}
\newtheorem{lemma}[theorem]{Lemma}

\newcommand{\fsync}{\textsc{$\mathcal{F}$sync}}
\newcommand{\ssync}{\textsc{$\mathcal{S}$sync}}
\newcommand{\async}{\textsc{$\mathcal{A}$sync}}

\def\Look/{\texttt{Look}}
\def\Compute/{\texttt{Compute}}
\def\Move/{\texttt{Move}}
\def\LCM/{\texttt{LCM}}

\def\Gathering/{\textsc{Gathering}}
\def\ChainFormation/{\textsc{Chain-For\-ma\-tion}}
\def\MaxChainFormation/{\textsc{Max-Chain-Forma\-tion}}
\def\MaxLineFormation/{\textsc{Max-Line-Formation}}
\def\UniformCircleFormation/{\textsc{Uniform-Circle-Formation}}

\def\Oblot/{\ensuremath{\mathcal{OBLOT}}}
\def\Luminous/{\ensuremath{\mathcal{LUMI}}}

\newcommand{\half}{\ensuremath{\frac{1}{2}}}
\newcommand{\quarter}{\ensuremath{\frac{1}{4}}}

\usepackage{authblk}

\title{The Max-Line-Formation Problem
\thanks{This work was partially supported by the German Research Foundation (DFG) under the project number 453112019; ME 872/14-1.}
}

\author[]{Jannik Castenow}
\author[]{Thorsten G\"otte}
\author[]{Till Knollmann}
\author[]{Friedhelm Meyer auf der Heide}
\affil[]{Heinz Nixdorf Institute and Department of Computer Science\\
	Paderborn University,
 \{janniksu, thgoette, tillk, fmadh\}@mail.upb.de}

\date{}

\makeatletter
\newenvironment{breakablealgorithm}
{% \begin{breakablealgorithm}
	\begin{center}
		\refstepcounter{algorithm}% New algorithm
		\hrule height.8pt depth0pt \kern2pt% \@fs@pre for \@fs@ruled
		\renewcommand{\caption}[2][\relax]{% Make a new \caption
			{\raggedright\textbf{\fname@algorithm~\thealgorithm} ##2\par}%
			\ifx\relax##1\relax % #1 is \relax
			\addcontentsline{loa}{algorithm}{\protect\numberline{\thealgorithm}##2}%
			\else % #1 is not \relax
			\addcontentsline{loa}{algorithm}{\protect\numberline{\thealgorithm}##1}%
			\fi
			\kern2pt\hrule\kern2pt
		}
	}{% \end{breakablealgorithm}
		\kern2pt\hrule\relax% \@fs@post for \@fs@ruled
	\end{center}
}
\makeatother
\begin{document}

\maketitle

\begin{abstract}
We consider $n$ robots with \emph{limited visibility}: each robot can observe other robots only up to a constant distance denoted as the \emph{viewing range}.
The robots operate in discrete rounds  that are either fully synchronous (\fsync{}) or semi-synchronized (\ssync{}).
Most previously studied formation problems in this setting seek to bring the robots closer together (e.g., \Gathering/ or \ChainFormation/).
In this work, we introduce the \MaxLineFormation/ problem, which has a contrary goal:
to arrange the robots on a straight line of maximal length.

First, we prove that the problem is impossible to solve by robots with a constant sized \emph{circular} viewing range.
The impossibility holds under comparably strong assumptions: robots that agree on \emph{both} axes of their local coordinate systems in \fsync{}.
On the positive side, we show that the problem is solvable by robots with a constant \emph{square} viewing range, i.e., the robots can observe other robots that lie within a constant-sized square centered at their position.
In this case, the robots need to agree on only \emph{one} axis of their local coordinate systems.
We derive two algorithms:
the first algorithm considers oblivious robots (\Oblot/) and converges to the optimal configuration in time $\mathcal{O}(n^2 \cdot \log (n/\varepsilon))$ under the \ssync{} scheduler ($\varepsilon$ is a convergence parameter).
The other algorithm makes use of locally visible lights (\Luminous/).
It is designed for the \fsync{} scheduler and can solve the problem exactly in optimal time $\Theta(n)$.
We also argue how a combination of the two algorithms can solve the \MaxLineFormation/ exactly in time $\mathcal{O}(n^2)$ under the \ssync{} scheduler with the help of the \Luminous/ model.

Afterward, we show that both the algorithmic and the analysis techniques can also be applied to the \Gathering/ and \ChainFormation/ problem:
we introduce an algorithm with a reduced viewing range for \Gathering/ and give new and improved runtime bounds for the \ChainFormation/ problem.
\end{abstract}

%%%%%%%%%%%%%%%%%%%%%%%%%%%%%%%%%%%%%%%%%%%%%%%%%%%%%%%%%%%%%%%%%%%%%%%%%%%%%%%%%%%%%%%%%%%%%%%%%%%
%%%%%%%%%%%%%%%%%%%%%%%%%%%%%%%%%%%%%%%%%%%%%%%%%%%%%%%%%%%%%%%%%%%%%%%%%%%%%%%%%%%%%%%%%%%%%%%%%%%

\section{Introduction}

Robot formation tasks aim to arrange $n$ mobile robots in a specific formation.
The robots are modeled as points in the Euclidean plane, and usually, the robot capabilities are very restricted.
Robots are assumed to be externally \emph{identical} (all robots have the same appearance), \emph{anonymous} (no identifiers), \emph{autonomous} (no central control) and \emph{homogeneous} (all robots execute the same algorithm).
Furthermore, the robots operate in discrete rounds denoted as \LCM/ cycles.
Each \LCM/ cycle consists of three operations: \Look/, \Compute/ and \Move/.
During the \Look/ operation, each robot takes a snapshot of its surroundings.
Afterward, the robot computes a target point during \Compute/ and finally moves there in the \Move/ operation.
With the additional assumptions that robots are \emph{silent} (no communication) and \emph{oblivious} (no memory of previous \LCM/ cycles), this is known as the \Oblot/ model~\cite{DBLP:series/lncs/FlocchiniPS19}.
The \Luminous/ model \cite{DBLP:series/lncs/LunaV19}, on the contrary, does not demand the robots to be silent and oblivious.
Instead, robots are equipped with a light that nearby robots (as well as the robot itself) can perceive.
The light can have different colors, and thus, the robots obtain a constant-sized memory and can communicate state information to their neighbors.
In addition to these core features, both models have a variety of freedom in some other assumptions; for instance, the \LCM/ cycles might be fully synchronous (\fsync{}), semi-synchronous (\ssync) or completely asynchronous (\async).
All schedulers are assumed to be fair such that each robot can execute its \LCM/ cycle infinitely often.
Time is measured in \emph{epochs}, i.e., the smallest number of rounds such that each robot has executed its \LCM/ cycle at least once.
In \fsync{} an epoch is equal to one round.

Our focus lies on robots with limited visibility, i.e., each robot cannot perceive the entire swarm but only nearby robots.
The terms \emph{connectivity range} and \emph{viewing range} are distinguished (see e.g., \cite{DBLP:conf/sss/CastenowKKH20,DBLP:conf/sss/PoudelS17}).
Robots are connected to all robots up to a distance equal to their connectivity range and can see all robots within their viewing range (the viewing range is at least as large as the connectivity range).
Initial configurations are connected w.r.t.\ the connectivity range and algorithms typically maintain this connectivity. The larger viewing range enhances the local information of the robots.
Additionally, viewing and connectivity range can be \emph{circular} or \emph{square}.
More precisely, a circular connectivity range of $c$ means that a robot is connected to all robots in the distance at most $c$ (all neighbors lie within the circle of radius $c$ around the robot).
In contrast, the square connectivity range of $sc$ connects a robot $r$ to all other robots located within an axis aligned $2sc \times 2sc $-sized square centered at $r$.
Similarly, circular and square viewing ranges are defined.
In many applications, the connectivity and the viewing range are identical.
The literature especially focusing on the runtime of formation algorithms
often benefits from a viewing range that is larger than the connectivity range, see e.g., \cite{DBLP:conf/ipps/AbshoffC0JH16,DBLP:journals/tcs/CastenowFHJH20,DBLP:conf/sss/CastenowH0KH20,DBLP:conf/sss/PoudelS17}.

Typical well-studied benchmark problems for robots with limited visibility are the
\Gathering/ and the \ChainFormation/ problem.
\Gathering/ demands the robots to gather at a single, not predefined, position.
\ChainFormation/ considers a chain of robots between two stationary outer robots: each inner robot has two identifiable neighbors (the neighborhoods are predefined and fixed). The goal is to arrange the robots on the line segment connecting the outer robots.
Both \Gathering/ and \ChainFormation/ can be characterized as \emph{contracting}: the robots move closer together.
Much less is known about formation tasks for robots with limited visibility that aim to achieve a contrary goal: to \emph{expand} the robots' positions.
One example is the \UniformCircleFormation/ problem in which $n$ robots are to move such that their positions form a regular polygon \cite{DBLP:conf/icdcit/DuttaCDM12,DBLP:conf/icdcit/MondalC20}.
Another, very recent example and the main inspiration for this work is the \MaxChainFormation/ problem~\cite{DBLP:conf/sss/CastenowKKH20}.
The \MaxChainFormation/ problem is a variant of the \ChainFormation/ problem.
The difference is that \MaxChainFormation/ gives the outer robots the ability to move.
The new goal is to transform the chain of robots with connectivity and viewing range $c$ into a straight line of length $(n-1) \cdot c$.

In this work, we introduce the \MaxLineFormation/ problem.
The goal is similar to the \MaxChainFormation/ problem: to move the robots with connectivity range $c$ such that their positions form a straight line of length $(n-1) \cdot c$.
The difference is that \MaxLineFormation/ does not consider predefined chain neighborhoods.
Instead, robots can observe the positions of all robots within their viewing range and do not have any fixed neighbors.
We analyze under which robot capabilities the problem is solvable, derive algorithms, and analyze their runtime.

\paragraph{Related Work} \label{section:relatedWork}

Due to space constraints, we focus on robots that operate in the \LCM/ model and results about \Gathering/, \ChainFormation/ and \MaxChainFormation/ with a particular focus on research that considers a runtime analysis of the proposed algorithms.
For a very recent and comprehensive overview of different robot formation algorithms, we refer the reader to \cite{DBLP:series/lncs/11340}.
Oblivious and disoriented robots (\Oblot/), can solve \Gathering/ in $\mathcal{O}(n^2)$ rounds (\fsync) with the \textsc{GTC} algorithm.
\textsc{GTC} moves robots in each round towards the center of the smallest enclosing circle of their neighborhood \cite{DBLP:journals/trob/AndoOSY99,DBLP:conf/spaa/DegenerKLHPW11}.
\textsc{GTC} achieves the currently best-known runtime for disoriented and oblivious robots in the Euclidean plane.
%	The same runtime can be achieved in three dimensions by moving robots towards the center of the smallest enclosing ball \cite{DBLP:conf/sirocco/BraunCH20}.
Faster algorithms for disoriented robots could so far only be designed under the \Luminous/ model.
There are two algorithms for robots located on a two-dimensional grid \cite{DBLP:conf/ipps/AbshoffC0JH16,DBLP:conf/spaa/Cord-Landwehr0J16}.
Another algorithm for robots in the Euclidean plane that are connected in a closed chain topology \cite{DBLP:conf/sss/CastenowH0KH20} exists.
When assuming the \Oblot/ model and one axis agreement, an asymptotically optimal algorithm with runtime $\mathcal{O}(\Delta)$ has been introduced in \cite{DBLP:conf/sss/PoudelS17}.
The algorithm assumes a square connectivity range of $1$ and a circular viewing range of $\sqrt{10}$.
Notably, the algorithm even works with the same runtime guarantees under the \async{} scheduler.

\ChainFormation/ has been initially introduced in \cite{DBLP:conf/ifip10/DyniaKLH06}.
The authors introduce the \textsc{GTM} algorithm that moves each robot to the midpoint between its neighbors.
For the \fsync{} scheduler, a runtime of $\mathcal{O}(n^2 \cdot \log (n/\varepsilon))$ rounds has been proven.
Later on, an almost matching lower bound (for the algorithm) of $\Omega(n^2 \cdot\log(1/\varepsilon))$ has been derived \cite{DBLP:conf/spaa/KlingH11}.
Algorithms with stronger assumptions, e.g., the \Luminous/ model, are able to achieve better runtimes \cite{DBLP:conf/spaa/DyniaKHS07,DBLP:journals/tcs/KutylowskiH09}.

Very recently, the \MaxChainFormation/ problem has been introduced~\cite{DBLP:conf/sss/CastenowKKH20}.
Started in one-dimensional configurations, the \textsc{Max-GTM} algorithm has a runtime of $\mathcal{O}(n^2 \cdot \log(n/\varepsilon))$ and $\Omega(n^2 \cdot \log(1/\varepsilon))$ rounds under the \fsync{} scheduler.
However, a specific class of input configurations does not converge to the optimal configuration.
For two-dimensional configurations, only a convergence result is known.
Additionally, for \Gathering/, \ChainFormation/ and \MaxChainFormation/, it is known that the problems can be solved optimally in a continuous time model \cite{DBLP:conf/sss/CastenowKKH20,DBLP:journals/topc/DegenerKKH15}.

\paragraph{Our Contribution}

We introduce the \MaxLineFormation/ problem.
The goal is to arrange $n$ robots with connectivity range $c$ on a straight line of length $(n-1) \cdot c$.
We start with an impossibility result and prove that there are initial configurations for which the problem cannot be solved deterministically by robots with constant sized \emph{circular} viewing and connectivity ranges.
In addition, also no algorithm that converges to the optimal solution can exist for these configurations.
The impossibility result even holds under strong assumptions: fully synchronized robots (\fsync{}) that agree on \emph{both} axes of their local coordinate systems.
On the positive side, we show that the problem becomes solvable for robots with identical \emph{square} connectivity and viewing ranges.
While \emph{square} connectivity and viewing ranges already have been proven to be useful to derive an efficient \Gathering/ algorithm \cite{DBLP:conf/sss/PoudelS17}, the \MaxLineFormation/ is the first known problem that can be solved under \emph{square} viewing ranges but not under \emph{circular} viewing ranges.
Our algorithms require the robots to agree on only \emph{one} axis of their local coordinate systems.
We introduce two algorithms:
The first algorithm considers the \Oblot/ model and converges to the optimal solution in $\mathcal{O}(n^2 \cdot \log (n/\varepsilon))$ epochs under the \ssync{} scheduler.
The analysis idea is based on the \emph{sample variance} of time \emph{inhomogeneous} Markov chains (a concept similar to the \emph{mixing time} of the time homogeneous case) inspired by \cite{DBLP:journals/tac/NedicOOT09}.
Afterward, we show that enhancing the robots with the \Luminous/ model allows us to derive an improved algorithm, i.e., the algorithm solves the problem exactly while simultaneously improving the runtime.
The algorithm considers the \fsync{} scheduler and solves the problem in $\Theta(n)$ epochs.
The runtime is asymptotically optimal.
Additionally, we argue that, with some additional synchronization, a combination of the two algorithms can solve the problem exactly with the help of lights in $\mathcal{O}(n^2)$ epochs under the \ssync{} scheduler.
Due to space constraints, only the high-level idea of the \ssync{} algorithm is contained in this version of the paper.

Our results compare to the \textsc{Max-GTM} algorithm for \MaxChainFormation/ (which has the same goal but considers predefined and fixed chain neighborhoods) problem as follows: our runtime of the \Oblot/ algorithm holds under the \ssync{} scheduler.
For \textsc{Max-GTM}, only runtimes in \fsync{} are known \cite{DBLP:conf/sss/CastenowKKH20}.
Additionally, our results about \MaxLineFormation/ hold for \emph{every} input configuration in which robots have distinct initial positions.
For \textsc{Max-GTM}, only a convergence result for a large class of configurations is known.
Additionally, certain classes of configurations do not converge to the optimal configuration \cite{DBLP:conf/sss/CastenowKKH20}.

Moreover, we identify an interesting relation to \Gathering/ and \ChainFormation/.
We first show that we can apply the main algorithmic idea of the $\Theta(n)$ algorithm to the \Gathering/ problem.
More precisely, we derive an algorithm for the \Oblot/ model that solves \Gathering/ of $n$ robots that agree on one axis of their local coordinate systems in $\Theta(\Delta)$ epochs under the \fsync{} scheduler, where $\Delta$ denotes the maximum distance of two robots in the initial configuration. \footnote{$\Omega(\Delta)$ is a trivial lower bound since at least one of the robots forming the diameter $\Delta$ must cover a distance of at least $\frac{\Delta}{2}$ to obtain \Gathering/. Since the robots have limited visibility, this requires $\Omega(\Delta)$ rounds.}
The algorithm uses a square viewing and connectivity range of $1$.
Up to now, the best-known algorithm achieving the same runtime uses a square connectivity range of $1$ and a circular viewing range of $\sqrt{10}$ \cite{DBLP:conf/sss/PoudelS17}.
Thus, our algorithm closes the gap between viewing and connectivity range.
Furthermore, we show how the analysis technique of the first algorithm (based on time inhomogeneous Markov chains) can also be applied for the \ChainFormation/ problem.
In this context, \emph{disoriented} robots (no agreement on the local coordinate systems) that are connected in a chain topology are assumed as well as a circular connectivity range and viewing range of $1$.
We prove that the \textsc{GTM} algorithm \cite{DBLP:journals/tcs/CohenP08,DBLP:conf/spaa/DyniaKHS07}, in which each robot moves to the midpoint between its two direct neighbors in every round, converges to the optimal configuration in $\mathcal{O}(n^2 \cdot \log (n/\varepsilon))$ epochs assuming the \ssync{} scheduler.
For one-dimensional configurations (all robots are initially collinear) this is a significant improvement over the so far best known runtime bound of $\mathcal{O}(n^5 \cdot \log(n/\varepsilon))$ epochs for this algorithm~\cite{DBLP:journals/tcs/CohenP08}.
For two-dimensional configurations, our result is the first runtime bound for the \ChainFormation/ problem derived for the \ssync{} scheduler.

\section{Model \& Notation}

\textbf{Time Model}
Robots operate in discrete LCM (\Look/, \Compute/, \Move/) cycles, denoted as \emph{rounds}.
Each robot takes a snapshot of its neighborhood during \Look/, computes a target point in \Compute/, and moves to this point in \Move/.
We assume a \emph{rigid} movement, robots always reach their target points during \Move/.
%In the context of limited visibility, rigid movement is a natural assumption as target points always lie in the visibility range of a robot.
The timing of the executions of the LCM cycles is either fully synchronous (\fsync{}) or semi-synchronous (\ssync), i.e., the cycles are synchronous, but only a subset of all robots participates.
The executions are always fair: All robots execute their cycles infinitely often.
Time is measured in epochs, i.e., the smallest number of rounds until each robot processes one complete \LCM/ cycle.
We assume that the execution starts in round $t_0$ and denote the first round of the $k$-th epoch by $t_{e_k}$.
Thus, $t_{e_1} = t_0$.

\textbf{Robot Model}
We consider $n$ robots $r_1, \dots, \,r_{n}$ positioned in $\mathbb{R}^2$.
Initially, the robots are located at pairwise distinct locations \footnote{Otherwise, the problem is deterministically unsolvable since multiplicities cannot be resolved if the robots are activated simultaneously. }.
We assume a square connectivity and viewing range of $1$, i.e., two robots $r_i$ and $r_j$ are neighbors if $r_j$ is located inside of the $2 \times 2$-sized square centered at $r_i$ and vice versa.
Note that $1$ is only chosen for simplicity; it can be replaced by any constant $c$.
The neighborhood of a robot $r_i$ (the set of all visible robots) in round $t$ is denoted by $N_i(t)$.
The square connectivity graph in which two robots share an edge if they are neighbors is initially connected.
Robots are assumed to be \emph{transparent} and thus do not block the views between other robots.
Moreover, the robots agree on one axis of their local coordinate systems.
W.l.o.g.\, we assume that the robots agree on the $x$-axis.
Thus, the robots have a common understanding of left and right, while up and down can be inverted.
However, the robots agree on unit distance and can measure distances precisely.
When considering the \Oblot/ model, the robots are also silent and oblivious.

For one algorithm, we consider the \Luminous/ model.
Each robot is equipped with a constant number of lights $\ell_{1}, \dots, \ell_{k}$ with color sets $C_1, \dots , C_k$ and at every point in time, each light can have a single color out of its color set. \footnote{In the classical \Luminous/ model \cite{DBLP:series/lncs/LunaV19} each robot is equipped with a single light and color set. Our assumption of multiple lights and color sets can be transferred to the classical setting by choosing a single light with a color set of size at most $2^{\sum_{i=1}^{k} |C_i|}$ .}
Robots can perceive the lights of their neighbors during \Look/ and can manipulate their light during \Compute/.
Hence, if a robot $r_i$ decides to change its light color in round $t$, its neighbors can see this earliest in round $t+1$.

\textbf{Notation}
The position of a robot $r_j$ in round $t$ is denoted by $p_j(t) = (x_j(t), y_j(t))$ in a global coordinate system and by $p^{i}_j(t) = (x^{i}_j(t), y^{i}_j(t))$ in the local coordinate system of $r_i$.
Each robot lies in the center of its local coordinate system and thus $p^{i}_i(t) = (0,0)$.
For a robot $r_i$, $r^{i}_\ell(t)$ denotes the leftmost robot of its neighborhood in round $t$.
The position of $r^{i}_{\ell}(t)$ in the local coordinate system of $r_i$ in round $t$ is denoted by $p^{i}_\ell(t) = (x^{i}_\ell(t), y^{i}_\ell(t))$.
In case there are multiple such robots, $r^{i}_{\ell}(t)$ represents an arbitrary robot of all leftmost robots.
Similarly, $r^{i}_r(t)$ and $p^{i}_r(t)$ are defined for the rightmost neighbor.
Additionally, define $r^{i}_+(t)$ and $p^{i}_+(t)$ to be the closest neighbor above of $r_i$ and its position.
Analogously, $r^{i}_-(t)$ and $p^{i}_-(t)$ is defined as the closest neighbor below and its position.
In case no such robot exists, $r^{i}_+(t) = r_i$ and  $r^{i}_-(t) = r_i$.
For a vector $v$, we denote by $\widehat{v}$ the normalized vector $\frac{1}{\|v\|}v$.

\textbf{Problem Statement}
\MaxLineFormation/ demands to move $n$ robots with connectivity range $c$ such that their positions form a straight line of length $(n-1) \cdot c$.
We say that an $(1-\varepsilon)$-approximation of the optimal configuration is reached if the positions form a straight line of length at least $(1-\varepsilon) \cdot (n-1) \cdot c$.
During the entire execution of an algorithm, the connectivity graph has to remain connected.

\section{Impossibility Result \& Intuition about Square Ranges} \label{section:impossibility}

This section proves that \MaxLineFormation/ is unsolvable with constant-sized circular viewing and connectivity ranges.
Afterward, we give an intuition on how square ranges circumvent the impossibility.

\subsection{Impossibility with Circular Ranges} \label{section:impossibilityCirc}
\begin{restatable}{theorem}{firstTheorem} \label{theorem:mainTheorem}
	In the \Oblot/ model, for every constant sized \emph{circular} connectivity and viewing range, there exists an initial configuration with robots located at distinct positions such that the \MaxLineFormation/ problem is unsolvable.
	Furthermore, no convergence algorithm can exist for these configurations.
	This holds for robots that agree on both axes of their local coordinate systems and the \fsync{} scheduler.
\end{restatable}

\begin{figure}[htbp]
	\begin{minipage}[t]{0.5\textwidth}
		\centering
		\includegraphics[width=0.65\textwidth]{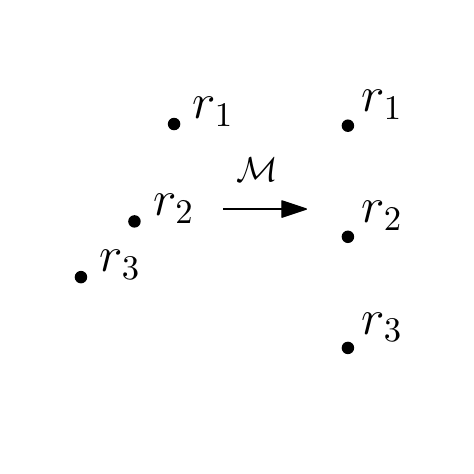}
		\caption{The config. $C_1$ transformed by $\mathcal{M}$.}
		\label{fig:impossibility}
	\end{minipage}
	$\,\,\,$
	\begin{minipage}[t]{0.5\textwidth}
		\centering
		\includegraphics[width=0.65\textwidth]{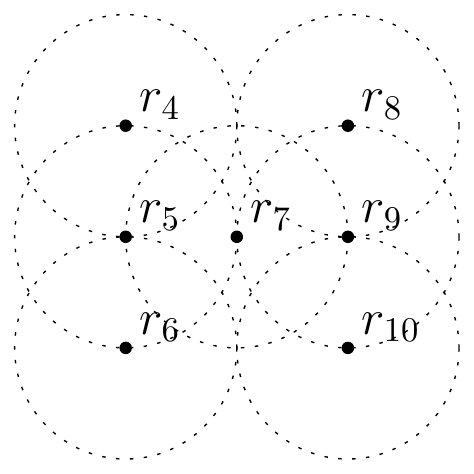}
		\caption{The configuration $C_2$.}
		\label{fig:impossibilitySquare}
	\end{minipage}

\end{figure}

\begin{proof}
	Initially, we assume identical viewing and connectivity ranges.
	The arguments for viewing ranges that are larger than the connectivity range are analogous and can be found in \Cref{section:completeImpossibility}.
	Thus, we assume a circular viewing and connectivity range of $c$.
	We prove the claim by contradiction.
	We assume that there is an algorithm $\mathcal{M}$ that is able to solve the \MaxLineFormation/ problem.
	Next, we derive a combination of $2$ initial configurations $C_1$ and $C_2$ and prove that if $\mathcal{M}$ is able to solve the problem starting in $C_1$, it cannot solve it starting in $C_2$.
	The configuration $C_1$ consists of three robots $r_1$, $r_2$ and $r_3$ at arbitrary (connected) positions.
	Since $\mathcal{M}$ is able to solve the problem, there is a time step $t_f$ such that the \MaxLineFormation/ problem is solved.
	W.l.o.g.\ we assume that $r_1$ and $r_3$ are located at the end of the line and $p_1(t_f), p_2(t_f)$ and $p_3(t_f)$ form a line parallel to the $y$-axis (otherwise we could rename the robots and rotate the following configuration $C_2$ accordingly).
	More precisely, $p_1(t_f) - p_2(t_f) = p_2(t_f) - p_3(t_f) = (0, c)$.
	See \Cref{fig:impossibility} for a depiction of the effects of $\mathcal{M}$ started in $C_1$.

	The configuration $C_2$ consists of $7$ robots, $r_4, \dots, r_{10}$ located at the following positions in a global coordinate system (not known to the robots):
	$p_4(t) =(-c,c) , p_5(t) =(-c,0), p_6(t) = (-c,-c), p_7(t) = (0,0), p_8(t) = (c,c), p_9(t) = (c,0),$ and  $p_{10}(t) = (c,-c)$.
	See \Cref{fig:impossibilitySquare} for a visualization of the configuration.
	In $C_2$, $r_4$ can only see $r_5$ and is located in distance $c$ of $r_5$.
	Moreover, it holds $p_4(t) - p_5(t) = p_1(t_f) - p_2(t_f)$ and $\|p_4(t) - p_5(t)\| = c$.
	Thus, $\mathcal{M}$ is not allowed to move $r_4$ since $\mathcal{M}$ cannot distinguish $r_1$ in configuration $C_1$ after time $t_f$ and $r_4$ in configuration $C_2$.
	By similar arguments, $\mathcal{M}$ is also not allowed to move $r_6,r_8$ and $r_{10}$.
	Hence, the only remaining robots that could be moved by $\mathcal{M}$ are $r_5$, $r_7$ and $r_9$.
	However, also these robots are not allowed to move.
	Consider the robot $r_5$ which is located in maximum distance to $r_4$, $r_6$ and $r_7$.
	No matter where $r_5$ moves, it loses the connectivity to either $r_4$ or $r_6$ as these robots remain at their position.
	The same arguments hold for $r_7$ and $r_9$.
	It follows that $\mathcal{M}$ cannot solve the problem $C_2$, which contradicts the assumption. \qed
\end{proof}

\subsection{Intuition about Square Ranges}

Next, we argue why the proof of \Cref{theorem:mainTheorem} does not hold when considering square viewing and connectivity ranges.
Assume that the algorithm $\mathcal{M}$ transforms the configuration $C_1$ into a line that is parallel to the $y$-axis.
Then, also the configuration $C_2$ is aligned with the $y$-axis.
Still, the robots $r_4, r_6, r_8$ and $r_{10}$ are not allowed to move.
The robots $r_5$ and $r_9$, however, gain the possibility to move horizontally.
More precisely, $r_5$ is allowed to move to the right (a distance of at most $1$) without losing the connectivity to $r_4$ and $r_6$ since the complete line segment connecting $r_5$ and $r_7$ is contained in the square viewing range of both $r_4$ and $r_6$.
Similarly, $r_9$ can move to the left.
See \Cref{fig:impossibilitySquareRanges} for a depiction of $C_2$ with square ranges instead of circular ones.
Consequently, an algorithm solving the \MaxLineFormation/ with the help of square ranges should arrange the robots on a line parallel to the $y$-axis.
The square ranges are only beneficial in case the local coordinate systems have the same orientation.
In case the robots are disoriented, the same impossibility result of \Cref{section:impossibilityCirc} also holds with square ranges.

%\begin{figure}
%	\centering
%
%
%
%\end{figure}

\begin{figure}[htbp]
	\begin{minipage}[c]{0.67\textwidth}
		\centering
		\includegraphics[width=0.4\textwidth]{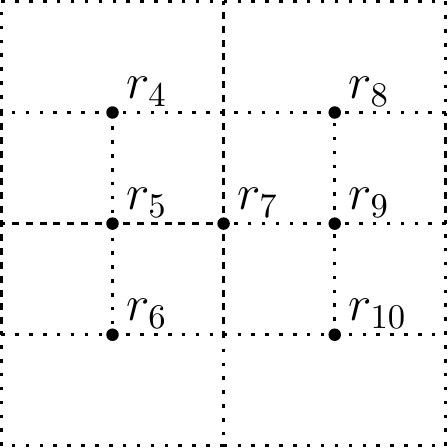}
	\end{minipage}\hfill
	\begin{minipage}[c]{0.3\textwidth}
		\caption{The configuration $C_2$ with square ranges instead of circular ones.}	\label{fig:impossibilitySquareRanges}
	\end{minipage}
\end{figure}

\section{\Oblot/ Algorithm} \label{section:oblot}

Based on the results of \Cref{section:impossibility}, \MaxLineFormation/ is unsolvable with \emph{circular} viewing and connectivity ranges.
In this section, we show that equipping the robots with \emph{square} connectivity and viewing ranges allows us to design an algorithm that converges to the optimal solution.
More precisely, we give an algorithm that converges to the optimal configuration assuming the \Oblot/ model and a square viewing and connectivity range of $1$.
%\Cref{fig:impossibilitySquare} depicts the configuration of \Cref{fig:impossibility} (for which we proved the impossibility result with circular ranges) with square viewing and connectivity ranges.
%One can see that the square shape of the ranges allows the robots $r_2$ and $r_6$ to move horizontally without losing the connectivity to their neighbors.
%The algorithm uses this observation by moving the rightmost robots to the left until all robots are collinear.

\subsection{Intuition}

The algorithm works in two phases.
In the first phase, the positions of all robots are arranged on a straight line parallel to the $y$-axis.
Afterward, the line is stretched in the second phase.
Since the robots are oblivious and have limited visibility, robots cannot distinguish the phases and act upon their local view.
Nevertheless, we will show that there is a time $t'$ such that all robots have joined the second phase and will remain there for the rest of the execution.

\textbf{Phase 1:}
A robot $r_i$ whose neighborhood has not yet formed a line parallel to the $y$-axis moves only if its position is rightmost in its neighborhood.
Then, $r_i$ moves horizontally to the $x$-coordinate of its leftmost neighbor.
If another robot already occupies this position, $r_i$ executes a slight vertical movement into the positive (from its local view) $y$-direction to avoid a collision.
Collisions have to be avoided as they cannot be resolved deterministically.
More precisely, if the robot is located topmost in its neighborhood, it moves a constant distance upwards.
If the robot is not topmost, it determines the value $y^{i}_{min}$, the $y$-coordinate of its closest neighbor to the top.
Afterwards, it moves $\frac{1}{3}y^{i}_{min}$ upwards.
The factor of $\frac{1}{3}$ is essential since the robot with $y$-coordinate $y^{i}_{min}$ might do the same movement while having a different understanding of up and down.
Hence, a collision of the two robots is avoided.

\textbf{Phase 2:} In the second phase, all robots are located on the same line parallel to the $y$-axis, which can be seen as a particular case of the \MaxChainFormation/ problem.
Thus, the robots execute the \textsc{Max-GTM} algorithm designed for \MaxChainFormation/ \cite{DBLP:conf/sss/CastenowKKH20}: each inner robot (robots that have neighbors in each direction) move to the midpoint between their closest northern and their closest southern neighbor.
Outer robots (at the end of the line) have to stretch the line and move as far as possible away from their closest neighbor without losing connectivity.
Concretely, outer robots move as follows.
Let $r_1$ be an outer robot and $r_2$ its closest neighbor and $v(t) = p_1(t) - p_2(t)$.
Then, $r_1$ imagines a virtual robot $r_v$ at the position $p_v(t) = p_1(t) + \widehat{v}(t)$ and moves to $\half p_v(t) + \half p_2(t)$.
\subsection{Algorithm}

We define the following set of possibly colliding robots.
For a robot $r_i$, define $C_i(t) = \{r_j \in N_i(t) | x^i_j(t) = 0  \textnormal{ or } x^i_j(t)  = x^{i}_{\ell}(t)\}$.
Now, $r^{i}_{min} \in C_i(t)$ is the robot with minimal $y^{i}_{min}(t)$ among all robots with $y^{i}_{min}(t) > 0$.
Thus, $r^{i}_{min}$ represents the robot lying above of $r_i$ (from $r_i$'s view) that has the smallest $y$-coordinate among all robots in $C_i(t)$.
If no such robot exists, define $y^{i}_{min} = \frac{1}{10}$.
\Cref{algorithm:Collisions} describes the movement of a robot $r_i$.

\begin{algorithm}
	\caption{ \Oblot/ \MaxLineFormation/ }\label{algorithm:Collisions} \begin{algorithmic}[1]
		\If{$x^{i}_{r}(t) = 0$ and $x^{i}_{\ell}(t)  < 0$} \Comment{Check if $r_i$ is rightmost but not leftmost}
		\If{no robot is located on $(x^{i}_\ell(t), 0)$}
		\State $p_i(t+1) \gets (x^{i}_\ell(t),0)$ \Comment{$r_i$ can move safely to the left}
		\Else
		\State $p_i(t+1) \gets (x^{i}_\ell(t), \frac{1}{3} \cdot y^{i}_{min})$ \Comment{$r_i$ avoids a collision with a vertical movement}
		\EndIf
		\Else

		\If{$x^{i}_{r}(t) = 0$ and $x^{i}_{\ell}(t)  = 0$} \Comment{Check if neighbors are collinear}
		\If{$y^{i}_{+}(t) = 0$ and $y^{i}_{-}(t) < 0$} \Comment{Check if $r_i$ is top most}
		\State $v_-(t) \gets p^{i}_-(t) - p_i(t)$; $p_v(t) \gets p_i(t) - \widehat{v}_-(t)$ \Comment{Position of virtual robot}
		\State $p_i(t+1) \gets \frac{1}{2}p_-(t) + \frac{1}{2}p_v(t)$
		\ElsIf{$y^{i}_{+}(t) > 0$ and $y^{i}_{-}(t)  = 0$} \Comment{Check if $r_i$ is bottom most }
		\State $v_+(t) \gets p^{i}_+(t) - p_i(t)$; $p_v(t) \gets p_i(t) - \widehat{v}_+(t)$\Comment{Position of virtual robot}
		\State $p_i(t+1) \gets \frac{1}{2}p_+(t) + \frac{1}{2}p_v(t)$
		\Else
		\State $p_i(t+1) \gets \frac{1}{2} p_-(t) +\frac{1}{2} p_+(t) $
		\EndIf
		\EndIf
		\EndIf
		\State $r_i$ moves to $p_i(t+1)$
	\end{algorithmic}
\end{algorithm}

\subsection{Analysis} \label{section:maxssyncAnalysis}

Next, we introduce the analysis idea to prove the main theorem (\Cref{theorem:mainTheoremOblot}) about the \Oblot/ algorithm.
Due to space constraints, all proofs are deferred to \Cref{section:maxLineOblotProofs}.

\begin{theorem} \label{theorem:mainTheoremOblot}
	For every $0  < \varepsilon < 1$, after $\mathcal{O}(n^2 \cdot \log \left(n/\varepsilon\right))$ epochs, the robots have formed a line of length at least $(1-\varepsilon) \cdot (n-1)$.
\end{theorem}

First, we argue that the first phase of the algorithm ends after $\mathcal{O}(n^2)$ rounds.

\begin{restatable}{lemma}{firstLemma}\label{lemma:firstLemma}
	After $\mathcal{O}(n^2)$ epochs, all robots are located on distinct positions on the same vertical line parallel to the $y$-axis.
	Moreover, the configuration is connected.
\end{restatable}

Now, we can assume that the first phase is completed, and thus all robots are located on the same vertical line.
W.l.o.g., we rename the robots such that $y_1(t) \leq y_2(t) \leq ... \leq y_n(t)$.
Moreover, define $w_1(t) = 1$ and $w_i(t) = y_i(t) - y_{i-1}(t)$ for $2 \leq i \leq n$.
In addition, define $z_i(t) = (w_i(t)-w_1(t))$.
The algorithm is designed such that $\lim_{t \rightarrow \infty} w_i(t) = 1$ for all $i$.
To analyze this behavior, we consider the following function:
$\Phi(t) = \sum_{i=2}^{n} z_i(t)^2.$
The function $\Phi(t)$ is also known as the \emph{sample variance} \cite{DBLP:journals/tac/NedicOOT09}.
The name comes from a relation to time inhomogeneous Markov chains.
Although the algorithm is deterministic, the behavior of the vectors $w_i(t)$ can be interpreted as a time inhomogeneous Markov Chain.
%The chain has an absorbing state (since $w_1(t) = 1$ for all $t$).
The main course of our analysis is based on \cite{DBLP:journals/tac/NedicOOT09}, where the authors analyzed a similar behavior in the context of the distributed averaging consensus problem.
In this problem, there are $n$ agents, each having a numerical opinion.
Every round, an agent gets to know some other opinions and updates its opinion to the average.
Our application has one important difference: the values $w_i(t)$ do not average but converge to the fixed value $w_1(t)$.
Hence, many parts of the proof in \cite{DBLP:journals/tac/NedicOOT09} have to be reworked and adapted to our application.
First, we derive a bound on the change of $\Phi(t)$ between two epochs.
Define $w_{\pi_1}(t_{e_k}), w_{\pi_2}(t_{e_k}),$ $\dots,  w_{\pi_n}(t_{e_k})$ to be the values $w_i(t_{e_k})$ sorted from largest to smallest with ties broken arbitrarily.

\begin{restatable}{lemma}{thirdLemma}\label{lemma:sortedBound}
	For any epoch $k$,
	$\Phi(t_{e_k}) - \Phi(t_{e_{k+1}}) \geq \frac{1}{4} \sum_{i=1}^{n-1} \left(w_{\pi_i}(t_{e_k})-w_{\pi_{i+1}}(t_{e_k})\right)^2$.
\end{restatable}

Based on \Cref{lemma:sortedBound}, a lower bound on the relative change is derived.

\begin{restatable}{lemma}{fourthLemma}\label{lemma:potentialRatio}
	Suppose that $\Phi(t_{e_k}) > 0$.
	Then, $\frac{\Phi(t_{e_k}) -\Phi(t_{e_{k+1}})}{\Phi(t_{e_k})} \geq \frac{1}{8 n^2}$.

\end{restatable}

A combination of \Cref{lemma:sortedBound,lemma:potentialRatio} yields the statement of \Cref{theorem:mainTheoremOblot}.

\section{\Luminous/ Algorithms} \label{section:luminousAlgorithms}

In this section, we derive an algorithm that solves \MaxLineFormation/ \emph{exactly} with the help of the \Luminous/ model under the \fsync{} scheduler (\Cref{section:lumiFsync}).
The algorithm achieves an asymptotically optimal runtime of $\Theta(n)$ rounds.
Additionally, in \Cref{section:lumiSSync}, we give an intuition about how a synchronization technique in combination with the \Oblot/ (\Cref{section:oblot}) and the \fsync{} algorithm (\Cref{section:lumiFsync}) is able to solve \MaxLineFormation/ exactly under the \ssync{} scheduler in $\mathcal{O}(n^2)$ epochs.

\subsection{Fast Algorithm for the \fsync{} scheduler} \label{section:lumiFsync}

The algorithm(\Cref{algorithm:luminousFsync}) also works in two phases: In the first phase, all robots are arranged on a straight line parallel to the $y$-axis, and in the second phase, the line is stretched until it has maximal length.
Compared to the \Oblot/ algorithm (\Cref{section:oblot}), the algorithm uses different core ideas in both phases.
In the first phase, all robots (instead of only the rightmost ones of their neighborhood) move to the left without losing connectivity -- this is necessary to achieve a linear speedup of the first phase.
The second phase makes use of lights to implement a sequential movement denoted as a \emph{run} inspired by  \cite{DBLP:conf/ipps/AbshoffC0JH16,DBLP:conf/sss/CastenowH0KH20,DBLP:conf/spaa/Cord-Landwehr0J16,DBLP:journals/tcs/KutylowskiH09}.
For the sake of clarity and due to space constraints, we present a variant of the algorithm in which the robots still move to the left during the second phase.
More precisely, after a linear number of rounds, the first phase ends, and the robots form a line parallel to the $y$-axis that continuously moves a distance of $1$ to the left.
Simultaneously, the robots stretch the line until it has maximal length.
However, the line structure is always maintained such that \MaxLineFormation/ is solved finally and remains solved (although the line keeps moving to the left).
Moving continuously to the left can be removed from the algorithm with some additional effort; an intuition is given in \Cref{section:fsyncAlgoNoMoving}.

\textbf{Phase 1:}
\emph{All} robots move as far as possible to the left : each robot $r_i$ moves to the $x$-coordinate $x^{i}_r(t)-1$.
Again, collision avoidance has to be ensured.
While moving to $x^{i}_r(t)-1$, the robot $r_i$ could collide with every robot located on its local $x$-axis (since these robots potentially also want to move to the $x$-coordinate $x^i_r(t) -1$).
The robot $r_i$ executes a vertical movement to avoid a collision.
Based on the ordering of neighbors on the local $x$-axis, $r_i$ gets assigned a unique $y$-coordinate as follows:
Define  $Y_i(t) = \{r_j \in N_i(t) |\, y^i_j(t) = 0 \}$ and let $x_{\pi_1}(t), x_{\pi_2}(t), \dots, x_{\pi_{|Y_i(t)|}}(t)$ be the $x$-coordinates of robots in $Y_i(t)$ in increasing order.
Additionally, let $k_i(t) \in \{1, \dots, |Y_i(t)|\}$ denote the position of $x_i(t)$ in the sorted sequence $x_{\pi_1}(t), x_{\pi_2}(t), \dots, x_{\pi_{|Y_i(t)|}}(t)$.
Furthermore, define $y^i_{min}(t)$ to be the minimal $y^{i}_j(t)$ of all $y^{i}_j(t) > 0$ of robots $r_j \in N_i(t)$.
If no such robot exists, define $y^{i}_{min}(t) = \frac{1}{10}$ (any constant of size at most $1$ works).
Then, $r_i$ gets assigned the $y$-coordinate $\frac{k_i(t)-1}{|Y_i(t)|} \cdot \frac{1}{3} y^{i}_{min}(t)$.
The factor $\frac{k_i(t)-1}{|Y_i(t)|}$ is unique for every robot on the local $x$-axis and the factor of $\frac{1}{3}$ is needed to prevent a collision with other robots that execute the same collision avoidance.

\textbf{Phase 2:}
For the second phase, lights are used.
Assume w.l.o.g.\ that the robots are ordered along the $y$-axis, i.e., $y_1(t) \geq y_2(t)  \geq \dots \geq y_n(t)$.
The core idea is a sequential movement started at $r_1$ and $r_n$ implemented with lights.
Such a movement is called a \emph{run} \cite{DBLP:conf/ipps/AbshoffC0JH16,DBLP:conf/sss/CastenowH0KH20,DBLP:conf/spaa/Cord-Landwehr0J16,DBLP:journals/tcs/KutylowskiH09}.
Assume that a run starts in round $t$.
Then, only $r_1$ and $r_n$ move.
In round $t+1$, only $r_2$ and $r_{n-1}$ move and so on.
A new run is started every three rounds.

Runs are realized with lights as follows.
The first required light $\ell_c$ with color set $C_c = \{0,1,2\}$ is used as a round counter.
Every round, all robots increment their light $\ell_c$.
Whenever $\ell_c = 2$ holds, both $r_1$ and $r_n$ activate a light $\ell_{mov}$ with $C_{mov} = \{0,1\}$ (the light is either active or inactive).
An active light $\ell_{mov}$ enables the corresponding robot to move.
Thus, in the next round, it holds $\ell_c = 0$ and both $r_1$ and $r_n$ detect an active light $\ell_{mov}$.
Both $r_1$ and $r_n$ now execute a movement (see below).
Additionally, they deactivate the light $\ell_{mov}$ and activate a light $\ell_{prev}$ with color set $C_{prev} = \{0,1\}$ to remember the movement.
Simultaneously, the robots $r_2$ and $r_{n-1}$ observe a neighbor on the $y$-axis with active light $\ell_{mov}$ ($r_1$ and $r_n$).
Additionally, neither $r_2$ nor $r_{n-1}$ has activated $\ell_{prev}$.
Hence, the robots activate $\ell_{mov}$ to continue the run.
In the next round, $r_1$ and $r_n$ observe a neighbor with active light $\ell_{mov}$ but do not activate their own light $\ell_{mov}$ since $\ell_{prev}$ is active.
Doing so ensures that the run keeps a fixed direction along the line.

Robots that have a run (the light $\ell_{mov}$ is active) move as follows.
In case $r_1$ has a run and not $r_2$ ($n > 2$), $r_1$ moves in distance $1$ vertically away from $r_2$.
More formally, $p_1(t+1) = (x^{1}_r(t) -1, -\frac{y^{1}_2(t)}{|y^{1}_2(t)|})$ (remember that in this variant the robots move also in phase $2$ to the left).
Similar, $r_n$ moves away from $r_{n-1}$ in distance $1$.
In case a robot $r_i$ has a run that came from $r_{i-1}$ ($r_{i-1}$ has activated $\ell_{prev}$ and $r_i$ has activated $\ell_{mov}$) and $r_{i+1}$ does not have a run, $r_i$ moves in vertical distance $1$ away from $r{+1}$: $p_i(t+1) = (x^{i}_r(t)-1, -\frac{y^{i}_{i+1}(t)}{|y^{i}_{i+1}(t)|})$.
Lastly, in case two neighboring robots have a run, for instance $r_i$ and $r_{i+1}$ have activated $\ell_{mov}$ both move only a vertical distance of $\half$ away from each other: $p_i(t+1) = (x^{i}_r(t)-1, -\frac{y^{i}_{i+1}(t)}{2|y^{i}_{i+1}(t)|})$.
The handling of the lights and the corresponding movement is depicted in \Cref{fig:secondPhaseFsync}.

\begin{breakablealgorithm}
	\caption{\Luminous/ Algorithm \fsync{} executed from the local view of $r_i$} \label{algorithm:luminousFsync} \begin{algorithmic}[1]
		\If{all neighbors are located on the $y$-axis}
		\If{$r_i = r^{i}_+(t)$ or $r_i = r^{i}_-(t)$}
		\If{$\ell_{mov} = 1$} \Comment{$\ell_{mov} = 1$ implies $\ell_c = 0$}
		\State $\ell_{mov} \gets 0; \ell_{prev} \gets 1$
		\State $r_c \gets $ closest neighbor on $y$-axis
		\If{$r_c$ has activated $\ell_{mov}$} \Comment{Special case $n=2$}
		\State $p_i(t+1) \gets  (x^{i}_r(t)-1, - \frac{1}{2 \cdot |y_c(t)|} \cdot y_c(t))$ \Comment{Move distance of $\half$}
		\Else
		\State $p_i(t+1) \gets  (x^{i}_r(t)-1, - \frac{1}{|y_c(t)|} \cdot y_c(t))$ \Comment{Move distance of $1$ }
		\EndIf
		\Else
		\If{$\ell_c = 2$}
		\State $\ell_{mov} \gets 1$;
		\EndIf
		\State  $p_i(t+1) \gets (x^{i}_r(t)-1, 0))$
		\EndIf
		\Else
		\If{$\ell_{mov}=1$ }
		\State $\ell_{mov} \gets 0$, $\ell_{prev} \gets 1$
		\If{closest neighbor above and below have set $\ell_{mov} = 0$}
		\State $r_c \gets $ closest neighbor with $\ell_{prev} = 0$
		\State $p_i(t+1) \gets  (x^{i}_r(t)-1, - \frac{1}{|y_c(t)|} \cdot y_c(t))$
		\Else
		\State $r_c \gets$ neighbor with $\ell_{mov} = 1$
		\State $p_i(t+1) \gets  (x^{i}_r(t)-1, - \frac{1}{2 \cdot |y_c(t)|} \cdot y_c(t))$
		\EndIf
		\Else
		\If{closest neighbor above or below has set $\ell_{mov} = 1$}
		\If{$\ell_{prev} = 0$}
		\State $\ell_{mov} \gets 1$
		\Else
		\State $\ell_{prev} \gets 0$
		\EndIf
		\EndIf
		\State  $p_i(t+1) \gets (x^{i}_r(t)-1, 0))$
		\EndIf
		\EndIf
		\Else
		\State $\{\ell_{mov},\ell_{prev}\} \gets 0$ \Comment{Deactivate lights if neighborhood is not in phase $2$}
		\If{$|Y_i(t)| > 0$}
		\State $p_i(t+1) \gets (x^{i}_r(t)-1, \frac{k_i(t)-1}{|Y_i(t)|} \cdot \frac{1}{3} y^{i}_{min}(t))$
		\Else
		\State $p_i(t+1) \gets (x^{i}_r(t)-1, 0)$
		\EndIf
		\EndIf
		\State $\ell_c \gets \ell_c+1$
		\State $r_i$ moves to $p_i(t+1)$
	\end{algorithmic}
\end{breakablealgorithm}

%\begin{figure}[htb]
%	\centering
%	\includegraphics[width=0.7\textwidth]{figures/secondPhase.pdf}
%	\caption{A depiction of phase $2$. A square (cross) depicts a robot with active light $\ell_{mov}$ ($\ell_{prev}$).  Time proceeds from left to right. In the first line it holds $\ell_c = 2$ for all robots. In this round the top most and the bottom most robot activate $\ell_{mov}$.  In the next round ($\ell_c = 0$), these two robots move in distance $1$ of their neighbor (depicted by an arrow) and additionally deactivate $\ell_{mov}$ while activating $\ell_{prev}$. Afterward ($\ell_c = 1$) the next two robots have activated $\ell_{mov}$ and move in distance $1$ of their next neighbor and so on.}
%	\label{fig:secondPhaseFsync}
%\end{figure}

\begin{figure}[htbp]
	\begin{minipage}[c]{0.67\textwidth}
		\centering
		\includegraphics[width=\textwidth]{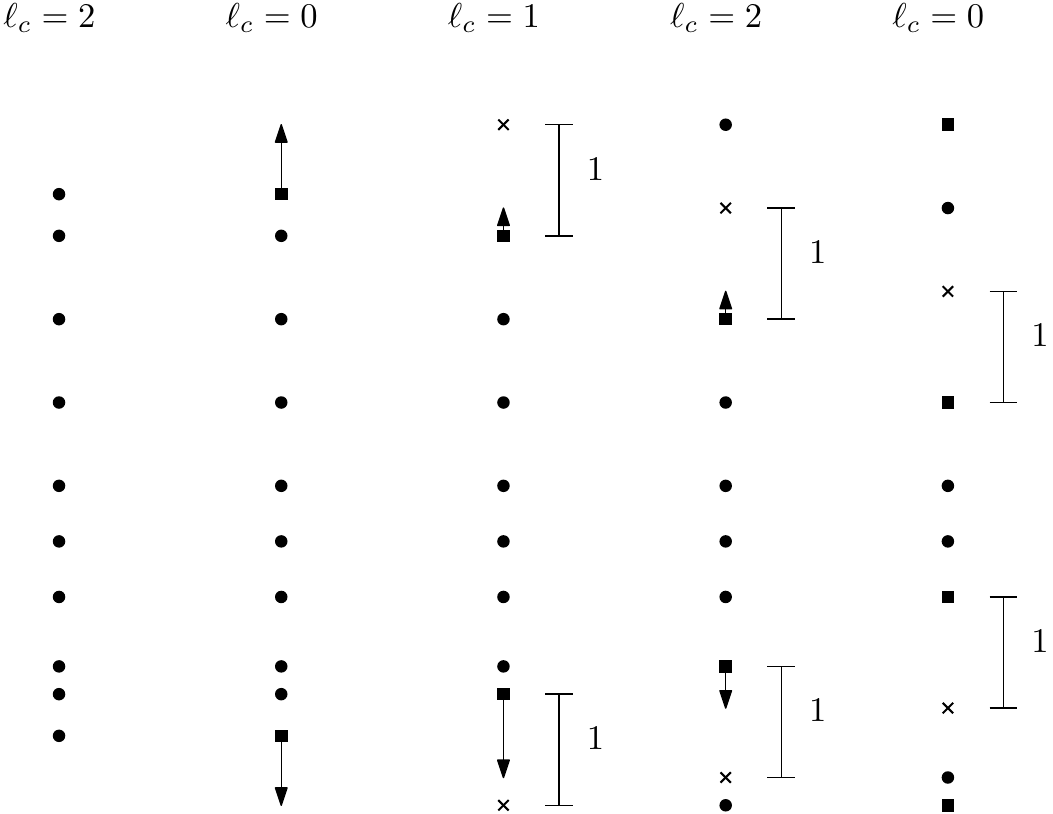}
	\end{minipage}\hfill
	\begin{minipage}[c]{0.3\textwidth}
		\caption{A square (cross) depicts a robot with active light $\ell_{mov}$ ($\ell_{prev}$).  Time proceeds from left to right. In the first line it holds $\ell_c = 2$ for all robots. In this round the top most and the bottom most robot activate $\ell_{mov}$.  In the next round ($\ell_c = 0$), these two robots move in distance $1$ of their neighbor (depicted by an arrow) and additionally deactivate $\ell_{mov}$ while activating $\ell_{prev}$. Afterward ($\ell_c = 1$) the next two robots with active light  $\ell_{mov}$ move in distance $1$ of their next neighbor.}
		\label{fig:secondPhaseFsync}
	\end{minipage}
\end{figure}

\textbf{Analysis:}
In the analysis (\Cref{section:appendixLumiProof}), it is proven that after a linear number of rounds, the first phase ends (and thus, the robots have formed a line parallel to the $y$-axis).
As a part of the proof, it is proven that no collisions occur, and the connectivity is always maintained.
Moreover, it is proven that as soon as phase $2$ is reached, the robots remain in phase $2$ (following from the algorithm's description).
Afterward, the runs of the second phase are analyzed.
The first run ensures that after $\mathcal{O}(n)$ rounds, the robots $r_{\lfloor n/2 \rfloor}$ and $r_{\lfloor n/2 \rfloor +1}$ have a vertical distance of $1$.
The second run ensures the same both  for $r_{\lfloor n/2 \rfloor -1}$ and $r_{\lfloor n/2 \rfloor}$ as well as  $r_{\lfloor n/2 \rfloor +1}$ and $r_{\lfloor n/2 \rfloor+2}$.
Hence, after $\mathcal{O}(n)$ runs, the line reaches maximal length.
Since each $3$ rounds, a new run is started, and each run proceeds one robot per round, the linear runtime follows.

\begin{restatable}{lemma}{sixthLemma}\label{lemma:fsyncCollinear}
	After $\mathcal{O}(n)$ epochs, all robots are located on distinct positions on the same vertical line parallel to the $y$-axis.
	Moreover, the configuration is connected.
\end{restatable}

\begin{restatable}{theorem}{thirdTheorem} \label{theorem:fsyncLumi}
	After $\mathcal{O}(n)$ epochs, the robots have solved \MaxLineFormation/.
\end{restatable}

The algorithm can be implemented in the classical \Luminous/ model with a single light having $9$ colors.
Observe that no robot ever activates the lights $\ell_{prev}$ and $\ell_{mov}$ at the same time.
Thus, for each robot, it always holds: either $\ell_{prev}$, $\ell_{mov}$ or none of both are activated.
Additionally, each robot counts rounds with the light $\ell_c$ requiring $3$ colors.
Hence, the total number of required colors is $9$: $3$ colors of $\ell_c$, each combined with $3$ possible cases for the lights $\ell_{mov}$ and $\ell_{prev}$.
\subsection{\ssync{} Scheduler} \label{section:lumiSSync}
The first phase of the \ssync{} algorithm is identical to the first phase of the \Oblot/ algorithm (\Cref{section:oblot}):
Each robot that is rightmost in its neighborhood moves horizontally to the $x$-coordinate of its leftmost neighbor.
In case this position is already occupied, a slight vertical movement is used to avoid collisions.
The main idea of the second phase is the sequential movement (run) of \Cref{algorithm:luminousFsync}.
Due to the \ssync{} scheduler, an additional synchronization procedure needs to be added.
In \fsync{}, a robot with active light $\ell_{mov}$ can always be sure that the neighbors observe and adapt the light.
Since only a subset of robots is active in every round in \ssync{}, the light $\ell_{mov}$ might not be seen, and thus, the run stops.
To overcome this, we add a synchronization done with the light $\ell_c$.
In contrast to the $\fsync{}$ algorithm, the robots do not increment the light in every round they become active.
Instead, each run gets associated with a color of the light $\ell_c$.
More precisely, the main idea is as follows.
Assume that the robots have already formed a line parallel to the $y$-axis.
Moreover, we rename the robots such that $y_1(t) \leq y_2(t) \leq \dots \leq y_n(t)$.
Additionally, assume the configuration is well-initialized, i.e. all robots have set $\ell_c = 0$.
We describe the procedure from the view of $r_1$, it works analogously for $r_n$.
We denote by $\ell_i(r_j)$ the color of $r_j$'s light $\ell_i$ in round $t$ (the time parameter is omitted for readability).
As soon as $r_1$ is activated, it observes $\ell_c(r_2) = \ell_{prev}(r_2) = \ell_{mov}(r_2) = 0$.
Then, $r_1$ activates $\ell_{mov}$.
As soon as $r_1$ wakes up again, it executes its movement (it moves in distance $1$ of $r_2$), deactivates $\ell_{mov}$, activates $\ell_{prev}$ and increments $\ell_c$ such that $\ell_c = 1$.
In the future, $r_1$ will only deactivate $\ell_{prev}$ in case it detects $\ell_c(r_2) = 1$ (indicating that $r_2$ has taken over the run).
Hence, as soon as $r_2$ is activated and detects $\ell_c(r_1) = \ell_{prev}(r_1) = 1$ and $\ell_c(r_3) = \ell_{prev}(r_3) = \ell_{mov}(r_3) = 0$, it will activate $\ell_{mov}$.
Upon its next activation, $r_2$ executes its movement, deactivates $\ell_{mov}$. activates $\ell_{prev}$ and increments $\ell_c$.
As soon as two neighboring robots have activated $\ell_{mov}$ both move in distance $\half$ away from each other and stop the run (exactly as in \Cref{algorithm:luminousFsync}).
This way, the runs proceed along the line.
To conclude, a robot $r_j$ only takes over a run from its neighbor $r_{j-1}$ in case $\ell_c(r_{j-1}) = \ell_c(r_j) + 1$.
Additionally, $r_j$ will only deactivate $\ell_{prev}$ as soon as $\ell_c(r_{j-1}) \geq \ell(r_j)$ and $\ell_c(r_{j+1}) = \ell_c(r_j)$.

Note that it might happen due to the limited visibility that some runs already start while the first phase is not completed.
Hence, at the beginning of phase $2$, not all robots might be initialized with the same color of the light $\ell_c$.
In case a robot detects such a violation (e.g., the next robot that should take over the light $\ell_{mov}$ has a larger value of $\ell_c$), the usual movement is not executed.
Instead, simply the light $\ell_c$ is incremented.
Hence, for each constant number of runs, the light of one more robot is well-initialized, and the algorithm adjusts the colors of the lights $\ell_c$ in a self-stabilizing manner.
All in all, the first phase has a runtime of $\mathcal{O}(n^2)$ epochs (\Cref{lemma:firstLemma}), the second phase is after $\mathcal{O}(n)$ epochs well-initialized (arguments above) and completed after additional $\mathcal{O}(n)$ epochs (\Cref{theorem:fsyncLumi}).
The runtime of $\mathcal{O}(n^2)$ epochs follows.

%
%\vspace*{-0.1cm}
%\begin{theorem}
%	\Luminous/ robots, that agree on one axis of their local coordinate systems and have a square viewing and connectivity range of $1$ can solve \MaxLineFormation/ exactly in $\mathcal{O}(n^2)$ epochs under the \ssync{} scheduler.
%\end{theorem}

\section{Relation to \Gathering/ and \ChainFormation/} \label{section:otherProblems}

Finally, we show that we can also apply the main ideas of our algorithms for the \MaxLineFormation/ problem in the context of \Gathering/ and \ChainFormation/.

\subsection{Gathering} \label{section:gathering}

We consider robots in the \Oblot/ model that agree on one axis of their local coordinate systems and operate under the \fsync{} scheduler.
%\Gathering/ demands to move the robots to gather at one (not predefined) position.
Define $\Delta$ to be the maximal distance of two robots in the initial configuration in round $t_0$.
Moreover, $\Delta_x$ denotes $\max_{i,j} |x_i(t_0) - x_j(t_0)|$ and analogously $\Delta_y$ denotes $\max_{i,j} |y_i(t_0) - y_j(t_0)|$.
Observe that $\Delta_x \in \mathcal{O}(\Delta)$ and $\Delta_y \in \mathcal{O}(\Delta)$.
The core idea of the \Gathering/ algorithm (\Cref{algorithm:gathering}) is as follows: to use the first phase of \Cref{algorithm:luminousFsync} presented in \Cref{section:lumiFsync} to arrange the robots on a vertical line fast.
In this phase, every robot moves as far as possible to the left.
While in \Cref{section:lumiFsync}, collisions have to be avoided, this is not necessary for \Gathering/ since collisions are desired to gather all robots on a single point.
In \Cref{section:lumiFsync} it has been proven that this phase requires $\mathcal{O}(n)$ epochs.
We show with a slightly more elaborate argument that this phase requires only $\mathcal{O}(\Delta)$ epochs.
The second phase squeezes the line to gather all robots and works as follows: robots at the end of the line move half the distance towards their farthest neighbor.
All other robots move to the midpoint between their farthest neighbor above and their farthest neighbor below.
The complete algorithm is contained in \Cref{algorithm:gathering}.
The following theorem states the $\mathcal{O}(\Delta)$ runtime, see \Cref{section:appendixSectionSixProofs} for a proof.

\begin{restatable}{theorem}{fourthTheorem}
	\Gathering/ of $n$ robots agreeing on one axis of their local coordinate systems in the \Oblot/ model can be solved in $\mathcal{O}(\Delta)$ epochs under the \fsync{} scheduler.
\end{restatable}

\begin{algorithm}
	\caption{ \Oblot/ \Gathering/ \ssync{} (executed if \Gathering/ not done)}\label{algorithm:gathering} \begin{algorithmic}[1]
		\If{all neighbors are located on the $y$-axis}
		\State $r_a^{i}(t) \gets $ farthest robot above ($r_i$ if no such robot exists)
		\State $r_b^{i}(t) \gets $ farthest robot below ($r_i$ if no such robot exists)
		\State $p_i(t+1) \gets (x^{i}_r(t)-1,\frac{1}{2} y_a^{i}(t) + \frac{1}{2} y_b^{i}(t))$
		\Else

		\State $p_i(t+1) \gets (x^{i}_r(t)-1,0)$

		\EndIf
		\State $r_i$ moves to $p_i(t+1)$
	\end{algorithmic}
\end{algorithm}
\setlength{\textfloatsep}{10pt}

\subsection{Chain-Formation} \label{section:chainFormation}
Lastly, we study the \ChainFormation/ problem that considers \emph{disoriented} robots.
Additionally, the robots are connected in a chain topology: there are $n+2$ robots $r_0, r_1, \dots, r_{n+1}$.
The robots $r_0$ and $r_{n+1}$, denoted as \emph{outer robots}, are stationary (they do not move).
Every other robot $r_i$ has exactly two chain neighbors: $r_{i-1}$ and $r_{i+1}$ whose positions it can always observe.
The robots have a circular connectivity and viewing range of $1$.
Define by $w_i(t)= (w^x_i(t), w^y_i(t)) = p_i(t) - p_{i-1}(t)$ the vectors along the chain and $L(t) = \sum_{i=1}^{n+1}\|w_i(t)\|$.
Additionally, $D = \|p_0(t)- p_{n+1}(t)\|$.
The goal of the \ChainFormation/ problem is to move the robots such that $L(t) = D$ and to distribute the robots uniformly along the line segment between $r_0$ and $r_{n+1}$.
W.l.o.g,, assume that $r_0$ is positioned in the origin of a global coordinate system and $r_{n+1}$ on the positive $x$-axis in distance $D$ to $r_0$.
Then, in the optimal configuration it holds $w_i(t) = w_{\infty} = \frac{D}{n+1}$ for $1 \leq i \leq n+1$.
We say that an $\varepsilon$-approximation of the optimal configuration is reached in case $\|w_i(t) - w_{\infty}\| \leq \varepsilon$ for all $1 \leq i \leq n+1$.

For the problem, the \textsc{GTM} algorithm has been introduced \cite{DBLP:journals/tcs/CohenP08,DBLP:conf/spaa/DyniaKHS07}.
The algorithm moves each robot in every round to the midpoint between its two direct neighbors.
%See \Cref{section:relatedWork} for more details about performance guarantees.
The \textsc{GTM} algorithm is very similar to the second phase of the \Oblot/ algorithm (\Cref{algorithm:Collisions}) presented in \Cref{section:oblot}.
Also, in \Cref{algorithm:Collisions}, robots that are not located at the end of the line move to the midpoint of their closest neighbors.
In \Cref{algorithm:Collisions}, however, the robots at the of the line are moving to stretch the line.
In contrast, the robots $r_0$ and $r_{n+1}$ of the \ChainFormation/ problem do not move.
Nevertheless, we  can apply a very similar analysis idea to the \textsc{GTM} algorithm:
We prove convergence independently for $w^{x}_i(t)$ and $w^y_i(t)$.
Since the arguments are identical, we concentrate on $w^{x}_i(t)$.
Define $\overline{x} = \frac{1}{n+1} \cdot \sum_{i=1}^{n+1} w^{x}_i(t)$.
Furthermore, define $z_i(t) = w^{x}_i(t) - \overline{x}$.
The analysis is based on the following function: $\Phi_2(t) = \sum_{i=1}^{n+1} z_i(t)^2$ that can be analyzed in most parts analogously  to $\Phi(t)$ in \Cref{section:oblot}.
See \Cref{section:chainFormationProofs} for a proof.

\begin{theorem} \label{theorem:chainFormationTheorem}
	For every $0 < \varepsilon < 1$,  \textsc{GTM} reaches an $\varepsilon$-approximation of the optimal configuration in $\mathcal{O}(n^2 \cdot \log(n/\varepsilon))$ epochs under the \ssync{} scheduler.
\end{theorem}

\section{Conclusion}
We have introduced the \MaxLineFormation/ problem and proven that the problem is impossible to solve with \emph{circular} viewing and connectivity ranges.
On the positive side, we have derived three algorithms for robots with square viewing and connectivity ranges.
Several open questions remain: is it possible to solve the \MaxLineFormation/ exactly when considering oblivious robots (\Oblot/)?
Is the derived runtime for the \Oblot/ model tight or can there be a more efficient algorithm?
The same question about lower bounds is also still open for the \ChainFormation/ and the \Gathering/ problem.
Can the problem be solved by disoriented robots (robots that do not agree on any axis)?
For the last question, certainly \emph{square} ranges do not help to solve the problem as the square ranges cannot be aligned according to a common axis.

\bibliographystyle{splncs04}
 \bibliography{maxLine}

\begin{thebibliography}{10}
\providecommand{\url}[1]{\texttt{#1}}
\providecommand{\urlprefix}{URL }
\providecommand{\doi}[1]{https://doi.org/#1}

\bibitem{DBLP:conf/ipps/AbshoffC0JH16}
Abshoff, S., Cord{-}Landwehr, A., Fischer, M., Jung, D., {Meyer auf der Heide},
  F.: Gathering a closed chain of robots on a grid. In: {IPDPS}. pp. 689--699.
  {IEEE} Computer Society (2016)

\bibitem{DBLP:journals/trob/AndoOSY99}
Ando, H., Oasa, Y., Suzuki, I., Yamashita, M.: Distributed memoryless point
  convergence algorithm for mobile robots with limited visibility. {IEEE}
  Trans. Robotics Autom.  \textbf{15}(5),  818--828 (1999)

\bibitem{DBLP:journals/tcs/CastenowFHJH20}
Castenow, J., Fischer, M., Harbig, J., Jung, D., {Meyer auf der Heide}, F.:
  Gathering anonymous, oblivious robots on a grid. T. C. S.  \textbf{815},
  289--309 (2020)

\bibitem{DBLP:conf/sss/CastenowH0KH20}
Castenow, J., Harbig, J., Jung, D., Knollmann, T., {Meyer auf der Heide}, F.:
  Brief announcement: Gathering in linear time: {A} closed chain of disoriented
  and luminous robots with limited visibility. In: {SSS}. LNCS, vol. 12514, pp.
  60--64. Springer (2020)

\bibitem{DBLP:conf/sss/CastenowKKH20}
Castenow, J., Kling, P., Knollmann, T., {Meyer auf der Heide}, F.: A discrete
  and continuous study of the max-chain-formation problem - slow down to speed
  up. In: {SSS}. LNCS, vol. 12514, pp. 65--80. Springer (2020)

\bibitem{DBLP:journals/tcs/CohenP08}
Cohen, R., Peleg, D.: Local spreading algorithms for autonomous robot systems.
  T. C. S.  \textbf{399}(1-2),  71--82 (2008)

\bibitem{DBLP:conf/spaa/Cord-Landwehr0J16}
Cord{-}Landwehr, A., Fischer, M., Jung, D., {Meyer auf der Heide}, F.:
  Asymptotically optimal gathering on a grid. In: {SPAA}. pp. 301--312. {ACM}
  (2016)

\bibitem{DBLP:journals/topc/DegenerKKH15}
Degener, B., Kempkes, B., Kling, P., {Meyer auf der Heide}, F.: Linear and
  competitive strategies for continuous robot formation problems. {ACM} Trans.
  Parallel Comput.  \textbf{2}(1),  2:1--2:18 (2015)

\bibitem{DBLP:conf/spaa/DegenerKLHPW11}
Degener, B., Kempkes, B., Langner, T., {Meyer auf der Heide}, F., Pietrzyk, P.,
  Wattenhofer, R.: A tight runtime bound for synchronous gathering of
  autonomous robots with limited visibility. In: {SPAA}. pp. 139--148. {ACM}
  (2011)

\bibitem{DBLP:series/lncs/LunaV19}
{Di Luna}, G.A., Viglietta, G.: Robots with lights. In: Distributed Computing
  by Mobile Entities, LNCS, vol. 11340, pp. 252--277. Springer (2019)

\bibitem{DBLP:conf/icdcit/DuttaCDM12}
Dutta, A., Chaudhuri, S.G., Datta, S., Mukhopadhyaya, K.: Circle formation by
  asynchronous fat robots with limited visibility. In: {ICDCIT}. LNCS,
  vol.~7154, pp. 83--93. Springer (2012)

\bibitem{DBLP:conf/ifip10/DyniaKLH06}
Dynia, M., Kutylowski, J., Lorek, P., {Meyer auf der Heide}, F.: Maintaining
  communication between an explorer and a base station. In: {BICC}. {IFIP},
  vol.~216, pp. 137--146. Springer (2006)

\bibitem{DBLP:conf/spaa/DyniaKHS07}
Dynia, M., Kutylowski, J., {Meyer auf der Heide}, F., Schrieb, J.: Local
  strategies for maintaining a chain of relay stations between an explorer and
  a base station. In: {SPAA}. pp. 260--269. {ACM} (2007)

\bibitem{DBLP:series/lncs/11340}
Flocchini, P., Prencipe, G., Santoro, N. (eds.): Distributed Computing by
  Mobile Entities, Current Research in Moving and Computing, LNCS, vol. 11340.
  Springer (2019)

\bibitem{DBLP:series/lncs/FlocchiniPS19}
Flocchini, P., Prencipe, G., Santoro, N.: Moving and computing models: Robots.
  In: Distributed Computing by Mobile Entities, LNCS, vol. 11340, pp. 3--14.
  Springer (2019)

\bibitem{DBLP:conf/spaa/KlingH11}
Kling, P., {Meyer auf der Heide}, F.: Convergence of local communication chain
  strategies via linear transformations. In: {SPAA}. pp. 159--166. {ACM} (2011)

\bibitem{DBLP:journals/tcs/KutylowskiH09}
Kutylowski, J., {Meyer auf der Heide}, F.: Optimal strategies for maintaining a
  chain of relays between an explorer and a base camp. T. C. S.
  \textbf{410}(36),  3391--3405 (2009)

\bibitem{DBLP:conf/icdcit/MondalC20}
Mondal, M., Chaudhuri, S.G.: Uniform circle formation by swarm robots under
  limited visibility. In: {ICDCIT}. LNCS, vol. 11969, pp. 420--428. Springer
  (2020)

\bibitem{DBLP:journals/tac/NedicOOT09}
Nedic, A., Olshevsky, A., Ozdaglar, A.E., Tsitsiklis, J.N.: On distributed
  averaging algorithms and quantization effects. {IEEE} Trans. Autom. Control.
  \textbf{54}(11),  2506--2517 (2009)

\bibitem{DBLP:conf/sss/PoudelS17}
Poudel, P., Sharma, G.: Universally optimal gathering under limited visibility.
  In: {SSS}. LNCS, vol. 10616, pp. 323--340. Springer (2017)

\end{thebibliography}

\newpage
\appendix

\appendix
\section{Complete proof of \Cref{section:impossibility}} \label{section:completeImpossibility}

\firstTheorem*

\begin{proof}
	Initially, we assume identical viewing and connectivity ranges.
	The arguments for viewing ranges that are larger than the connectivity range are analogous and can be found in \Cref{section:completeImpossibility}.
	Thus, we assume a circular viewing and connectivity range of $c$.
	We prove the claim by contradiction.
	We assume that there is an algorithm $\mathcal{M}$ that is able to solve the \MaxLineFormation/ problem.
	Next, we derive a combination of $2$ initial configurations $C_1$ and $C_2$ and prove that if $\mathcal{M}$ is able to solve the problem starting in $C_1$, it cannot solve it starting in $C_2$.
	The configuration $C_1$ consists of three robots $r_1$, $r_2$ and $r_3$ at arbitrary (connected) positions.
	Since $\mathcal{M}$ is able to solve the problem, there is a time step $t_f$ such that the \MaxLineFormation/ problem is solved.
	W.l.o.g.\ we assume that $r_1$ and $r_3$ are located at the end of the line and $p_1(t_f), p_2(t_f)$ and $p_3(t_f)$ form a line parallel to the $y$-axis (otherwise we could rename the robots and rotate the following configuration $C_2$ accordingly).
	More precisely, $p_1(t_f) - p_2(t_f) = p_2(t_f) - p_3(t_f) = (0, c)$.
	See \Cref{fig:impossibility} for a depiction of the effects of $\mathcal{M}$ started in $C_1$.

	The configuration $C_2$ consists of $7$ robots, $r_4, \dots, r_{10}$ located at the following positions in a global coordinate system (not known to the robots):
	$p_4(t) =(-c,c) , p_5(t) =(-c,0), p_6(t) = (-c,-c), p_7(t) = (0,0), p_8(t) = (c,c), p_9(t) = (c,0),$ and  $p_{10}(t) = (c,-c)$.
	See \Cref{fig:impossibilitySquare} for a visualization of the configuration.
	In $C_2$, $r_4$ can only see $r_5$ and is located in distance $c$ of $r_5$.
	Moreover, it holds $p_4(t) - p_5(t) = p_1(t_f) - p_2(t_f)$ and $\|p_4(t) - p_5(t)\| = c$.
	Thus, $\mathcal{M}$ is not allowed to move $r_4$ since $\mathcal{M}$ cannot distinguish $r_1$ in configuration $C_1$ after time $t_f$ and $r_4$ in configuration $C_2$.
	By similar arguments, $\mathcal{M}$ is also not allowed to move $r_6,r_8$ and $r_{10}$.
	Hence, the only remaining robots that could be moved by $\mathcal{M}$ are $r_5$, $r_7$ and $r_9$.
	However, also these robots are not allowed to move.
	Consider the robot $r_5$ which is located in maximum distance to $r_4$, $r_6$ and $r_7$.
	No matter where $r_5$ moves, it loses the connectivity to either $r_4$ or $r_6$ as these robots remain at their position.
	The same arguments hold for $r_7$ and $r_9$.
	It follows that $\mathcal{M}$ cannot solve the problem $C_2$, which contradicts the assumption. \qed

	Next, we consider a viewing range that is larger than the connectivity range but still a constant.
	W.l.o.g.\ we assume that there is a constant $\alpha  > 1$ such that the viewing range is of size $\alpha \cdot c$.
	The configuration is similar to before but the robots $r_4, r_6, r_7, r_8$ and $r_{10}$ are replaced by a line of $\lceil \alpha \rceil$ robots in maximum distance.
	More precisely, $r_4$ is replaced by $\lceil \alpha \rceil$ robots $r_{4,1}, r_{4,2}, \dots, r_{4,\lceil \alpha \rceil}$ with $p_{4,j}(t) = (-\lceil \alpha \rceil \cdot c, (\lceil \alpha \rceil - j +1) \cdot c)$.
	Similarly, $r_6, r_7, r_8$ and $r_{10}$ are replaced.
	Hence, in total $\lceil \alpha \rceil \cdot 5 +2$ robots are needed.
	See \Cref{fig:impossibilityLarger} for a visualization.
	The configuration is designed such that $r_{4,1}, r_{6,2}, r_{8_1}$ and $r_{10,2}$ are not allowed to move since the configuration looks like the final configuration from their point of view.
	All other robots are not allowed to move since their movement would disconnect the connectivity graph. \qed

	\begin{figure}[htbp]
		\centering
		\includegraphics[width=0.55\textwidth]{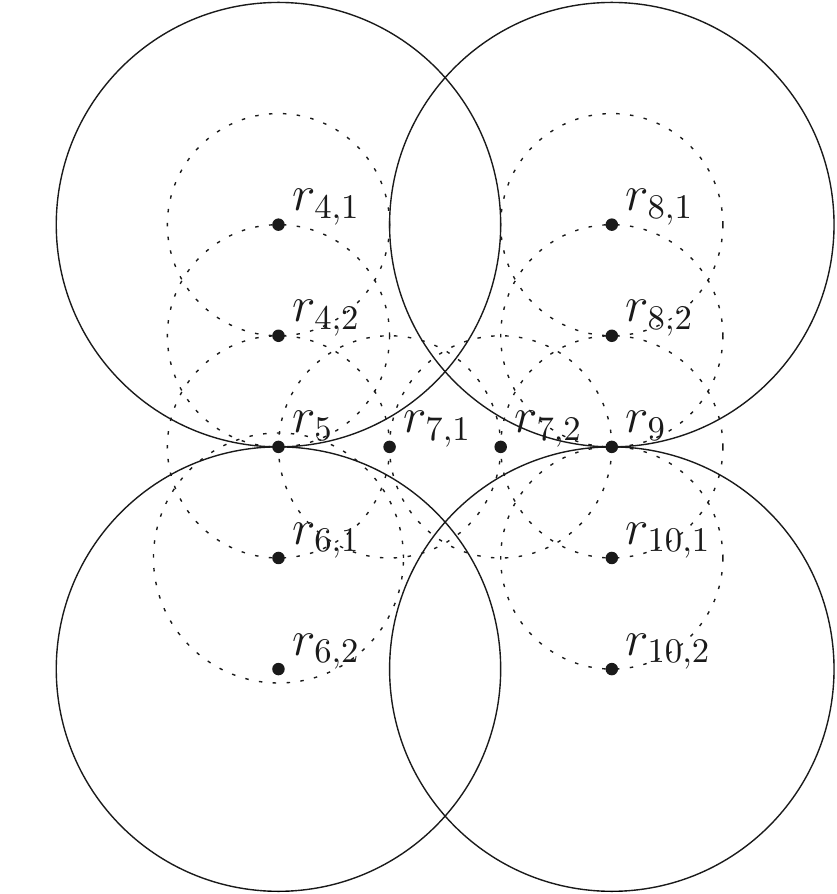}
		\caption{The configuration with $\alpha = 2$ is depicted. The dotted circles represent the connectivity ranges of the robots. The solid circles depict the viewing ranges of selected robots (other viewing ranges are left out for the sake of clarity).}
		\label{fig:impossibilityLarger}
	\end{figure}

\end{proof}
\clearpage
\newpage

\section{Omitted Proofs of \Cref{section:oblot}} \label{section:maxLineOblotProofs}

\firstLemma*

\begin{proof}
	Initially, at most $n$ distinct $x$-coordinates that are occupied by robots exist.
	In every epoch, at least one robot that occupies the rightmost $x$-position moves to the left as the configuration is connected.
	This movement does not create any new $x$-position as the robot moves to the $x$-coordinate of its leftmost neighbor.
	Additionally, no robot moves to the right.
	Hence, after at most $n$ epochs, no robot occupies the rightmost $x$-coordinate anymore.
	Thus, after $\mathcal{O}(n^2)$ epochs, all robots are located on the same vertical line by applying the same argument inductively.
	The connectivity and non-existence of collisions follow from the algorithm's description.
	\qed
\end{proof}

We define $\tau_i(t) =1$ if and only if $r_i$ is active in round $t$.
First of all, we derive formulas for the vectors $w_i(t+1)$.
For each vector, we have to consider $4$ cases: $\tau_i(t) = 1$ and $\tau_{i-1}(t) = 1$, $\tau_i(t) = 1$ and $\tau_{i-1}(t) = 0$, $\tau_i(t) = 0$ and $\tau_{i-1}(t) = 1$ and $\tau_i(t) = 0$ and $\tau_{i-1}(t) = 0$.
Furthermore, $\mu_{i}^{-}(t) = \tau_{i-1}(t) \cdot \left(\tau_{i-1}(t)-\tau_{i}(t)\right)$ and $\mu_{i}^{+}(t) = \tau_{i}(t) \cdot \left(\tau_{i}(t)-\tau_{i-1}(t)\right)$.
For the ease of notation, define $d_{i}^{-}(t) = \mu_{i}^{-}(t) \cdot \left(w_i(t)- w_{i-1}(t)\right)^2$, $d_i(t) = \tau_{i}(t) \cdot \tau_{i-1}(t) \cdot \left(w_{i-1}(t) - w_{i+1}(t)\right)^2$ and $d_{i}^{+}(t) = \mu_{i}^+(t) \cdot \left(w_i(t)- w_{i+1}(t)\right)^2$.
Observe that $d_i^{-}(t), d_i(t)$ and $d_i^{+}(t)$ are defined such that at most one of the three terms can be larger than $0$ (the other ones are equal to $0$).
Lastly, define $w_{n+1}(t) = w_1(t)$.

\begin{restatable}{lemma}{secondLemma} \label{lemma:potentialDifference}
	For any round $t$, it holds
	\begin{align*}
		\Phi(t+1) = \Phi(t) - \frac{1}{4}\sum_{i=2}^{n} d_{i}^{-}(t) + d_i(t) + d_i^{+}(t).
	\end{align*}
\end{restatable}

\begin{proof}
	Consider a vector $w_i(t)$ with $2<i<n$.
	Next, we calculate $w_i(t+1)$.
	There are $4$ cases to consider: Case 1: $\tau_{i-1}(t) = 1$ and $\tau_{i}(t) = 1$, Case 2 and 3: $\tau_{i-1}(t) = 1$ and $\tau_{i}(t) = 0$ or vice versa and Case $4$: $\tau_{i-1}(t) = 0$ and $\tau_{i}(t) = 0$.
	The following formulas can be easily verified:
	\begin{enumerate}
		\item Case 1: $w_i(t+1) = \half w_{i-1}(t) + \half w_{i+1}(t)$
		\item Case 2: $w_i(t+1) =  \half w_{i-1}(t) + \half w_i(t)$
		\item Case 3: $w_i(t+1) =\half w_{i}(t) + \half w_{i+1}(t)$
		\item Case 4: $w_i(t+1) = w_i(t)$
	\end{enumerate}

	The formulas for the boundary vectors $w_2(t)$ and $w_n(t)$ are slightly different:
	Case 1: $\tau_{1}(t) = 1$, $\tau_{2}(t) = 1$, $\tau_{n-1}(t) = 1, \tau_{n}(t) = 1$, Case 2 and 3: $\tau_{1}(t) = 1$, $\tau_{2}(t) = 0, \tau_{n-1}(t) = 0$ and $\tau_{n}(t) =1$ or vice versa and Case $4$: $\tau_{1}(t) = 0, \tau_{2}(t) = 0, \tau_{n-1}(t) = 0$ and $\tau_{n}(t) = 0$.

	\begin{enumerate}
		\item Case 1: $w_2(t+1) = \half w_{1}(t) + \half w_{3}(t)$; $w_n(t+1) = \half w_{n-1}(t) + \half w_{n+1}(t)$
		\item Case 2: $w_2(t+1) =   \half w_{1}(t) + \half w_{2}(t)$; $w_n(t+1) = \half w_{n}(t) + \half w_{n+1}(t) $
		\item Case 3: $w_2(t+1) =\half w_{2}(t) + \half w_{3}(t)$; $w_n(t+1) = \half w_{n-1}(t) + \half w_{n}(t)$
		\item Case 4: $w_2(t+1) = w_2(t)$; $w_n(t+1) = w_n(t)$
	\end{enumerate}

	Next, we derive a formula for $z_i(t+1)^2$ for $2 < i < n$.
	Observe first $z_i(t)^2 = (w_i(t)-1)^2 = w_i(t)^2 - 2 \cdot w_i(t) + 1$.

	\begin{enumerate}
		\item Case 1: $z_i(t+1)^2 = (\half w_{i-1}(t) + \half w_{i+1}(t)-1)^2 = \quarter w_{i-1}(t)^2 + \quarter w_{i+1}(t)^2 + \frac{w_{i-1}(t) \cdot w_{i+1}(t)}{2} - w_{i-1}(t) - w_{i+1}(t) +1 $
		\item Case 2: $z_i(t+1)^2  = \quarter w_{i-1}(t)^2 + \quarter w_{i}(t)^2 + \frac{w_{i-1}(t) \cdot w_{i}(t)}{2} - w_{i-1}(t) - w_{i}(t) +1 $
		\item Case 3: $z_i(t+1)^2  =\quarter w_{i}(t)^2 + \quarter w_{i+1}(t)^2 + \frac{w_{i}(t) \cdot w_{i+1}(t)}{2} - w_{i}(t) - w_{i+1}(t) +1 $
		\item Case 4: $z_i(t+1)^2  = z_i(t)^2$
	\end{enumerate}

	Similar formulas can be derived for $z_2(t+1)$ and $z_n(t+1)$.
	Since at most one of the three terms $d^{-}_i(t),d^{+}_i(t)$ and $d_i(t)$ is positive, and each $w_i(t)$ occurs exactly twice in all $z_i(t+1)$'s, the lemma follows.
	\qed
\end{proof}

\thirdLemma*

\begin{proof}
	By \Cref{lemma:potentialDifference}, we obtain

	\begin{align*}
		\Phi(t_{e_k})-\Phi(t_{e_{k+1}}) \geq \quarter \cdot \sum_{t=t_{e_k}}^{t_{e_{k+1}}} \sum_{i=2}^{n} d_{i}^{-}(t) + d_i(t) + d_{i}^{+}(t).
	\end{align*}

	The first part of the proof deals with finding a lower bound for any $d_{i}^{-}(t) + d_i(t) + d_{i}^{+}(t)$ given that at least one of the terms is larger than $0$ (at most one of the three terms is positive).
	The lower bound, however, depends on the sorted sequence $w_{\pi_1}(t)$, $\dots w_{\pi_n}(t)$.
	Since we lose much structure due to the sorting, some definitions are needed.
	Let $\pi$ be the function that maps the indices of $w_1(t_{e_k}), \dots, w_n(t_{e_k})$ into the sorted sequence $w_{\pi_1}(t_{e_k}), \dots, w_{\pi_n}(t_{e_k})$ and $\pi^{-1}$ its inverse.
	More precisely, for instance $\pi(i) = \pi_f$ if and only if $w_i(t_{e_k}) = w_{\pi_f}(t_{e_k})$.
	Furthermore, define $\sigma_{m,i,j}(t) = 1$ if and only if
	one of the following cases is fulfilled:

	\begin{enumerate}
		\item $d_{\pi_m}(t) > 0$ and $\pi(\pi^{-1}(\pi_m) -1) =\pi_i$ and $\pi(\pi^{-1}(\pi_m)+1) = \pi_j$ or vice versa
		\item $d_{\pi_m}^{-}(t) > 0$ and $\pi(\pi^{-1}(\pi_m)-1) = \pi_i$ and $\pi(\pi^{-1}(\pi_m)) = \pi_j$ or vice versa
		\item $d_{\pi_m}^{+}(t) > 0$ and $\pi(\pi^{-1}(\pi_m)) = \pi_i$ and $\pi(\pi^{-1}(\pi_m)+1) = \pi_j$ or vice versa
	\end{enumerate}

	Due to the sorting, we lose the nice property that only neighboring $w_i(t)$'s are involved in $d_{\pi_m}(t)$, $d^{-}_{\pi_m}(t)$ and $d^{+}_{\pi_m}(t)$.
	For instance in case $d_{\pi_m}(t) > 0$ we cannot conclude that $w_{\pi_{m}-1}(t)$ and $w_{\pi_{m}+1}(t)$ ar involved.
	Thus, intuitively, $\sigma_{m,i,j}(t) = 1$ if and only if $d_{\pi_m}(t), d_{\pi_m}^{-}(t)$ or a $d_{\pi_m}^{+}(t)$ is larger than $0$ and both $w_{\pi_i}(t)$ and $w_{\pi_j}(t)$ are involved.

	Next, define $t_{\ell}$ ($1 \leq \ell \leq n$) to be the first round larger than or equal to $t_{e_k}$ such that there exists three indices $\pi_i, \pi_j$ and $\pi_m$ ($\pi_i \neq \pi_j$ but $\pi_i = \pi_m$ or $\pi_j = \pi_m$ might hold) with $\pi_i \leq \pi_{\ell} < \pi_j$ (or vice versa) and $\sigma_{m,i,j}(t) = 1$.
	In other words, $t_{\ell}$ denotes the first round in which the values $w_{\pi_1}(t_{e_k}), \dots, w_{\pi_{\ell}}(t_{e_k})$ and $w_{\pi_{\ell+1}}(t_{e_k}), \dots,  w_{\pi_n}(t_{e_k})$ influence each other.
	By influencing each other, we mean that $w_{\pi_m}(t+1) = \half w_{\pi_i}(t)  + \half w_{\pi_j}(t)$,
	since either $d_{\pi_m}(t) > 0, d^{-}_{\pi_{m}}(t) > 0$ or $d^{
		+}_{\pi_{m}}(t) > 0$.\footnote{In the context of averaging consensus each index $1, \dots, n$ corresponds a node in the graph. Thus, the index $\ell$ can be interpreted as a cut in the graph and the time $t_{\ell}$ as the first time with a communication across the cut represented by $\ell$.}

	For all $t \in \{t_{e_k}, \dots,  t_{e_{k+1}}\}$ let $L(t) = \{\ell \, | \, t_{\ell} = t\}$, i.e. $L(t)$ represents all indices $\ell$ at time $t$ such that the two sets $\{w_{\pi_1}(t_{e_k}), \dots w_{\pi_{\ell}}(t_{e_k})\}$ and $\{w_{\pi_{\ell+1}}(t_{e_k}),$ $\dots, w_{\pi_n}(t_{e_k})\}$ influence each other for the first time.

	Now, we define all pairs of indices $\pi_i,\pi_j$ at time $t$ such that there exists an $\pi_\ell$ with $\pi_i \leq \pi_\ell < \pi_j$ and $\sigma_{\ell,i,j}(t) = 1$: $C_\ell(t) = \{\{\pi_i,\pi_j\} \, | \, \pi_i \leq \pi_\ell < \pi_j \textnormal{ and } \sigma_{\ell,i,j}(t) =1\}.$
	Lastly, define for fixed $i,j$ and $t$:
	$F_{ij}(t) = \{\ell \in L(t) \, | \, \{\pi_i, \pi_j\} \in C_{\ell}(t)\}$.

	Fix some $\pi_i$ and $\pi_j$ with $\pi_i < \pi_j$ and a round $t$ such that $|F_{ij}(t)| > 0$.
	Let $F_{ij}(t) = \{\ell_1, \dots, \ell_k\}$ sorted in increasing order.
	Since $\ell_1 \in L(t)$, it holds by definition that there exists no round $t' \in [t_{e_k}, \dots, t]$ and an index $\pi_m $ with $\pi_i \leq \pi_m < \pi_j$ such $\sigma_{m,i,j}(t') = 1$.
	It follows $w_{\pi_i}(t) \geq w_{\pi_{\ell_1}}(t_{e_k})$ (since $w_{\pi_i}(t)$ was so far only influenced by elements of the set $w_{\pi_1}(t_{e_k}), \dots, w_{\pi_{\ell}}(t_{e_k})$ which are all larger or equal to $w_{\pi_{\ell}}(t_{e_k})$).
	Similarly, one can argue $w_{\pi_j}(t) \leq w_{\pi_{\ell_k+1}}(t_{e_k})$.
	Hence, we can conclude

	\begin{align*}
		w_{\pi_i}(t) - w_{\pi_j}(t) \geq w_{\pi_{\ell_1}}(t_{e_k}) - w_{\pi_{\ell_k+1}}(t_{e_k}) \geq \sum_{\pi_\ell \in F_{ij}(t)} w_{\pi_\ell}(t_{e_k}) -w_{\pi_{\ell+1}}(t_{e_k}).
	\end{align*}

	The last line directly leads to

	\begin{align*}
		(w_{\pi_i}(t) - w_{\pi_j}(t))^2 \geq \sum_{\pi_\ell \in F_{ij}(t)} \left(w_{\pi_\ell}(t_{e_k}) -w_{\pi_{\ell+1}}(t_{e_k})\right)^2.
	\end{align*}

	The second part of the proof now deals with finding a lower bound for $\sum_{i=2}^{n} d_{i}^{-}(t) + d_i(t) + d_{i}^{+}(t)$ in a fixed round $t$.
	\begin{align*}
		\sum_{i=2}^{n} d_{i}^{-}(t) + d_i(t) + d_{i}^{+}(t) &=  \sum_{\left(\pi_m,\pi_i, \pi_j\right): \sigma_{m,i,j}(t) = 1} \left(w_{\pi_i}(t) - w_{\pi_j}(t)\right)^2 \\
		&\geq  \sum_{\left(\pi_m,\pi_i, \pi_j\right): \sigma_{m,i,j}(t) = 1} \sum_{\pi_\ell \in F_{ij}(t)} \left(w_{\pi_\ell}(t_{e_k}) -w_{\pi_{\ell+1}}(t_{e_k})\right)^2  \\
		&\geq  \sum_{\pi_\ell \in L(t)}\left(w_{\pi_\ell}(t_{e_k}) -w_{\pi_{\ell+1}}(t_{e_k})\right)^2.
	\end{align*}

	Lastly, we plug all insights together to conclude the proof.

	\begin{align*}
		\Phi(t_{e_k})-\Phi(t_{e_{k+1}}) &\geq \quarter \cdot \sum_{t=t_{e_k}}^{t_{e_{k+1}}} \sum_{i=2}^{n} d_{i}^{-}(t) + d_i(t) + d_{i}^{+}(t) \\
		&\geq \quarter \cdot \sum_{t=t_{e_k}}^{t_{e_{k+1}}} \sum_{\pi_\ell\in L(t)}\left(w_{\pi_\ell}(t_{e_k}) -w_{\pi_{\ell+1}}(t_{e_k})\right)^2. \\
		&= \quarter \sum_{\pi_\ell=1}^{n-1} \left(w_{\pi_\ell}(t_{e_k}) -w_{\pi_{\ell+1}}(t_{e_k})\right)^2.
	\end{align*}

	The last line follows since each robot moves at least once per epoch.
	\qed
\end{proof}

\fourthLemma*

\begin{proof}
	\Cref{lemma:sortedBound} leads to

	\begin{align*}
		\frac{\Phi(t_{e_k}) -\Phi(t_{e_{k+1}})}{\Phi(t_{e_k})} &\geq \frac{1}{4} \frac{ \sum_{\pi_\ell=1}^{n-1} \left(w_{\pi_\ell}(t_{e_k}) -w_{\pi_{\ell+1}}(t_{e_k})\right)^2}{ \sum_{\pi_\ell=1}^{n} \left(w_{\pi_\ell}(t_{e_k}) -1\right)^2.} \\
		&= \frac{1}{4} \frac{ \sum_{\pi_\ell=1}^{n-1} \left(w_{\pi_\ell}(t_{e_k}) -w_{\pi_{\ell+1}}(t_{e_k})\right)^2}{\sum_{\pi_\ell=1}^{n} \left(w_{\pi_\ell}(t_{e_k}) -w_{\pi_1}(t_{e_k})\right)^2.} \\
	\end{align*}

	The second line follows since $w_1(t) = 1$ for all $t$ and thus $w_{\pi_1}(t_{e_k}) =1$.
	Observe that the right-hand side does not change if we multiply each $w_{\pi_i}(t_{e_k})$ with the same constant.
	Additionally, it also does not change if we add the same constant to each $w_{\pi_i}(t_{e_k})$.
	Hence, we can assume w.l.o.g.\  $\sum_{\pi_\ell=1}^{n} w_{\pi_\ell}(t_{e_k}) = 0$ and $\sum_{\pi_\ell=1}^{n} \left(w_{\pi_\ell}(t_{e_k}) -w_{\pi_1}(t_{e_k})\right)^2 = 1$ and obtain

	\begin{align*}
		\frac{\Phi(t_{e_k}) -\Phi(t_{e_{k+1}})}{\Phi(t_{e_k})} &\geq \quarter \min_{\substack{w_1 \geq w_2, \dots, \geq w_n  \\ \sum_{i}w_i = 0 \\\sum_{i}(w_i-w_1)^2 = 1}} \sum_{i=1}^{n-1} \left(w_i-w_{i+1}\right)^2
	\end{align*}

	The assumption $\sum_{i}\left(w_i-w_1\right)^2 = 1$ implies that the average value of all $\left(w_i-w_1\right)^2$ is $\frac{1}{n}$ and hence there is at least some $j$ with $|w_j-w_1| \geq \frac{1}{\sqrt{n}}$.
	As a consequence, either $|w_1| \geq \frac{1}{2 \cdot \sqrt{n}}$ or $|w_j| \geq \frac{1}{2 \cdot \sqrt{n}}.$
	W.l.o.g. assume $|w_1| \geq \frac{1}{2 \sqrt{n}}$ and moreover assume $w_1 > 0$.
	The case $w_1 < 0$ can be handled by multiplying each $w_i$ with $-1$ and sort the elements in descending order.

	Now define $u_i = w_i - w_{i+1}$ for $i < n$ and $u_n = 0$.
	It holds $u_i \geq 0$ for all $i$ and $\sum_{i=1}^{n} u_i = w_1 - w_n$.
	Since at least one $w_1 \geq \frac{1}{2 \sqrt{n}}$ and the $\sum_{i}w_i = 0$, we can conclude $u_n < 0$ and thus $\sum_{i}u_i \geq \frac{1}{2 \sqrt{n}}$.

	As a final step, we obtain

	\begin{align*}
		\frac{\Phi(t_{e_k}) -\Phi(t_{e_{k+1}})}{\Phi(t_{e_k})} &\geq \quarter \min_{\substack{u_i  \geq 0 , \sum_{i}u_i \geq 1/(2 \sqrt{n})}} \sum_{i=1}^{n-1} u_i^2.
	\end{align*}

	The solution of the minimization problem is $u_i = \frac{1}{2 \cdot n^{3/2}}$ for each $i$.

	Hence,
	\begin{align*}
		\frac{\Phi(t_{e_k}) -\Phi(t_{e_{k+1}})}{\Phi(t_{e_k})} &\geq  \quarter \cdot \frac{1}{2 \cdot n^2} = \frac{1}{8n^2}.
	\end{align*}
	\qed
\end{proof}

\begin{restatable}{lemma}{fifthLemma}
	After $\mathcal{O}\left(n^2 \cdot \log \left(n/\varepsilon\right)\right)$ rounds, it holds $\sum_{i=2}^{n} w_i(t) \geq (1-\varepsilon) \cdot (n-1)$.
\end{restatable}

\begin{proof}
	Fix any epoch $e_k$.
	By \Cref{lemma:potentialRatio}, we obtain $\Phi(t_{e_{k+1}}) \leq (1-\frac{1}{8n^2}) \cdot \Phi(t_{e_k})$ and thus $\Phi(t_{e_{k+x}}) \leq \left(1-\frac{1}{8n^2}\right)^x \cdot \Phi(t_{e_k})$.
	Observe that $(1-y)^x \leq e^{-y\cdot x}$ where $e$ denotes Euler's number.
	Thus, choosing $x \geq 8n^2 \cdot \ln(\frac{1}{r})$ yields $\Phi(t_{e_{k+x}}) \leq r\cdot \Phi(t_{e_k})$.
	Since $\Phi(t_{e_k}) < n-1$, $r \leq \frac{\varepsilon}{n-1}$ leads to  $\Phi(t_{e_{k+x}}) \leq \varepsilon$ and thus $\sum_{i=2}^{n} w_i(t) \geq (1-\varepsilon) \cdot (n-1)$.
	\qed
\end{proof}

\newpage
\section{Adjusted \fsync{} Algorithm} \label{section:fsyncAlgoNoMoving}
The \fsync{} algorithm presented in \Cref{section:lumiFsync} is -- to keep the pseudocode comprehensible -- designed  such that the robots still move to the left after \MaxLineFormation/ is already solved.
In this section, we explain how $3$ additional lights help to remove this behavior to design an algorithm that forms a stationary line.

Observe first that in \Cref{algorithm:luminousFsync}, two runs can only be located at two neighboring robots in case the algorithm is already in phase $2$.
Otherwise, at least one robot observes that its neighborhood is not yet aligned parallel to the $y$-axis, and the corresponding run is stopped (line 31 in \Cref{algorithm:luminousFsync}).
We use this observation as follows:
As soon as two runs meet at neighboring robots, the two robots activate a light $\ell_{final}$ to store this information.
Robots with an active light $\ell_{final}$ do not move to the left anymore.
Additionally, robots that  observe a robot in their neighborhood that has activated $\ell_{final}$, activate their own light $\ell_{final}$.
Hence, after $\mathcal{O}(n)$ rounds, all robots have activated $\ell_{final}$.
While propagating $\ell_{final}$, it might, however, happen that some robots move to the left while other robots remain stationary (due to the limited visibility).
See \Cref{fig:fsyncAlgoNonStationary} for a depiction of such a case.
To rebuild the line-shape again runs at robots with active light $\ell_{final}$ behave slightly different: the vertical movement is identical to before.
The horizontal movement changes: instead of moving to the left, a robot moves a distance of $1$ to the right if it is leftmost in its neighborhood, and there is at least one robot in a horizontal distance of $1$ to the right.
Finally, the robots align again on the initial line (before activating $\ell_{final}$) and \MaxLineFormation/ gets solved after $\mathcal{O}(n)$ rounds.

\begin{figure}[htbp]
	\centering
	\includegraphics[width=0.75\textwidth]{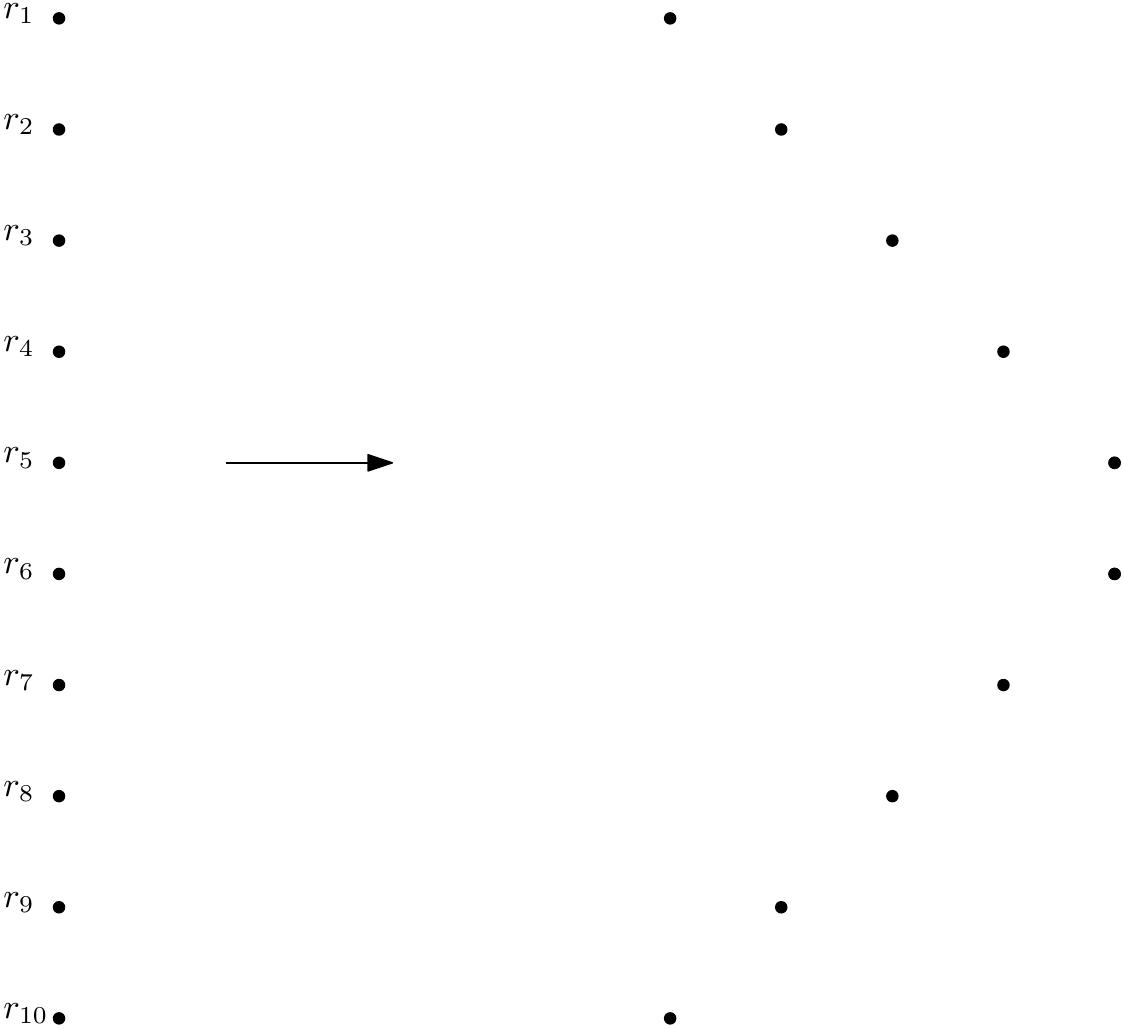}
	\caption{To the left, the robots are arranged on a straight line parallel to the $y$-axis. After some time, two runs meet in the middle at $r_5$ and $r_6$ that activate their light $\ell_{final}$. In the next round, $r_4$ and $r_7$ activate their light $\ell_{final}$ and so on. However, $r_1$, $r_2$, $r_3$, $r_4$, $r_7$, $r_8$, $r_9$ and $r_{10}$ move a distance of $1$ to the left before $r_4$ and $r_7$ become stationary. Similarly, $r_1, \dots, r_3$ and $r_8, \dots, r_{10}$ move a distance of $1$ further to the left than $r_3$ and $r_6$ and so on. The final configuration might look like it is depicted to the right.}
	\label{fig:fsyncAlgoNonStationary}
\end{figure}
\newpage
\section{Omitted Proofs of \Cref{section:luminousAlgorithms}} \label{section:appendixLumiProof}

\sixthLemma*

\begin{proof}
	The connectivity and collision avoidance follow directly from the algorithm description.
	To prove the linear runtime, define $x_{max}(t)$ to be the maximal $x$-coordinate in the global coordinate system.
	Furthermore, define $k_{max}(t)$ to be the number of robots with $x_i(t) \in (x_{max}(t)-1,x_{max}(t)]$.
	Observe that every robot $r_i$ with $x_i(t) \in (x_{max}(t)-1,x_{max}(t)]$ moves such that $x_i(t+1) \in (x_{max}(t)-2, x_{max}(t)-1]$.
	Since the configuration is always connected, there must have been a robot $r_j$ with $x_j(t) \in  (x_{max}(t)-2, x_{max}(t)-1]$ that cannot leave the interval.
	It follows $k_{max}(t+1) \geq k_{max}(t) +1$
	Thus, after $\mathcal{O}(n)$ rounds, all robots have $x$-coordinates in an interval of size at most $1$.
	Fix one round $t$ and assume that all robots have $x$-coordinates in an interval of size at most $1$.
	Consider the robot $r_{max}$ with globally maximal $x$-coordinate.
	None of its neighbors can see a robot that has a larger $x$-coordinate, and thus, all robots of $N_{max}(t)$ are collinear in round $t+1$.
	Furthermore, all of these robots still have the globally largest $x$-coordinate in round $t+1$.
	Now, consider the topmost robot $r_j \in N_{max}(t)$.
	It follows that all robots of $N_{j}(t+1)$ have a smaller or equal $x$-coordinate and cannot see any other robot with a larger $x$-coordinate.
	Thus, in round $t+2$ all robots in $N_{max}(t)$ and $N_{j}(t+1)$ are collinear.
	The same argument holds for the bottom-most robot.
	Applying the same argument inductively yields that all robots are collinear after $\mathcal{O}(n)$ rounds.
\end{proof}

\thirdTheorem*

\begin{proof}
	By \Cref{lemma:fsyncCollinear} all robots are collinear on a line parallel to the $y$-axis after $\mathcal{O}(n)$ epochs.
	It remains to prove the linear runtime until the optimal configuration is formed.
	Rename the robots such that $r_1$ is the topmost robot and $r_n$ is the bottom-most robot.
	Define $u_i(t)  = y_i(t) - y_{i-1}(t)$.
	Both $r_1$ and $r_n$ activate their light $\ell_{mov}$ every $3$ rounds and ensure $u_2(t) = u_n(t) = 1$.
	In round $t+1$ it holds $u_3(t+1) = u_{n-1}(t+1) = 1$ and so on.
	Assume $n$ to be even (the arguments for odd $n$ are analogous).
	After $\frac{n}{2}-1$ rounds, the movement meets at the two robots $r_{n/2}$ and $r_{n/2+1}$ and they move such that $y_{n/2+1}(t+\frac{n}{2}-1) = 1$ holds.
	The next two movements ensure $y_{n/2}(t+\frac{n}{2}+2) = y_{n/2+2}(t+\frac{n}{2}+2) = 1$ and so on.
	Since each $3$ rounds a new movement is started, the optimal configuration is reached in $\mathcal{O}(n)$ rounds.
\end{proof}

\newpage
\section{Omitted Proofs of \Cref{section:otherProblems}}\label{section:appendixSectionSixProofs}

\begin{restatable}{lemma}{seventhLemma}\label{lemma:gatheringCollinear}
	After $\mathcal{O}(\Delta)$ epochs, all robots are located on the same vertical line parallel to the $y$-axis.
	Moreover, the configuration is connected.
\end{restatable}

\begin{proof}
	Define by $x_{min}(t)$ the minimal $x$-coordinate of all robots.
	Next, define intervals $i_1(t) = [x_{min}(t)+1], i_2(t) = [x_{min}(t)+1, x_{min}(t)+2]$ and so on.
	Additionally, $R_j(t) := \{r_k \,|\, x_k(t) \in i_j(t)\}$ and $k(t)$ is the largest index $j$ of an interval such that $R_j(t) \neq \emptyset$.
	Observe that $k(t) \leq \Delta_x$.
	Fix a round $t_0$.
	By the same arguments as in the proof of \Cref{lemma:fsyncCollinear}, it holds $R_{k(t_0)}(t_0+1) = \emptyset$ (since all robots leave the rightmost interval).
	At the same time, it can happen that some of the robots in the leftmost interval $i_1(t_0)$ create a new interval to the left such that $i_1(t_0+1) \neq i_1(t_0)$.
	In total, however, at most $\Delta_y$ new intervals can be created since for every new interval, it must hold that the robots cannot observe the position of any of the previous new generated intervals.
	Since initially at most $\Delta_x$ intervals exist, we obtain a runtime of $\Delta_x+\Delta_y \in \mathcal{O}(\Delta)$ until all robots have an $x$-coordinate in the same interval.
	Now, consider the robot $r_{max}$ with globally maximal $x$-coordinate.
	None of its neighbors can see a robot that has a larger $x$-coordinate, and thus, all robots of $N_{max}(t)$ are collinear in round $t+1$.
	Furthermore, all of these robots still have the globally largest $x$-coordinate in round $t+1$.
	Next, consider the topmost robot $r_j \in N_{max}(t)$.
	It follows that all robots of $N_{j}(t+1)$ have a smaller or equal $x$-coordinate and cannot see any other robot with a larger $x$-coordinate.
	Thus, in round $t+2$ all robots in $N_{max}(t)$ and $N_{j}(t+1)$ are collinear.
	The same argument holds for the bottom-most robot.
	This case can occur at most $\Delta_y$ times.
	Afterward, all robots are collinear.
	The runtime of $\mathcal{O}(\Delta)$ follows.
\end{proof}

\fourthTheorem*

\begin{proof}
	According to \Cref{lemma:gatheringCollinear}, all robots are collinear after $\mathcal{O}(\Delta)$ rounds.
	Rename the robots such that $r_1$ is the topmost robot and $r_n$ the bottom-most robot.
	We prove exemplary for $r_1$ that it moves a constant distance toward $r_n$ every two rounds.
	The arguments for $r_n$ are analogous.
	Observe first that $r_1$ remains topmost because $r_1$ moves to the midpoint between its position and its farthest neighbor $r_f$.
	In case there is any robot between $r_1$ and $r_f$, this robot can also see both $r_1$ and $r_f$ and either moves to the same position as $r_1$ or can see a robot $r_{f'}$ that lies below of $r_f$.
	Similarly, $r_f$ and robots below of $r_f$ can only compute target points below the target point of $r_1$.
	Hence, $r_1$ remains the topmost robot.
	Now consider a round in which $r_1$ moves a distance of less than $\frac{1}{10}$ downwards.
	This implies that its farthest neighbor $r_f$ is in distance at most $\frac{1}{5}$.
	Since the configuration is connected, $r_f$ can see a robot $r_{f'}$ in distance at least $1$ of $r_1$.
	Hence, $r_1$ moves a distance of at least $\half$ downwards.
	Thus, the distance between $r_1$ and $r_f$ in round $t+1$ is at least $\half$.
	Hence, $r_1$ moves a constant distance (at least $\quarter$) in round $t+2$.
	Thus, every two rounds, $r_1$ and $r_n$ move at least a constant distance.
	Finally, they can see each other, and all robots gather in the next round.

\end{proof}

\subsection{Proof of \Cref{theorem:chainFormationTheorem}}\label{section:chainFormationProofs}

\begin{lemma} \label{lemma:potentialDifferenceChain}
	For any round $t$, it holds
	\begin{align*}
		\Phi(t+1) = \Phi(t) - \frac{1}{4}\sum_{i=2}^{n} d_{i}^{-}(t) + d_i(t) + d_i^{+}(t).
	\end{align*}
\end{lemma}

\begin{proof}
	Analogous to the proof of \Cref{lemma:potentialDifference}.
\end{proof}

Next, define $w_{\pi_1}(t_{e_k}), w_{\pi_2}(t_{e_k}),$ $\dots,  w_{\pi_n}(t_{e_k})$ to the values $w^{x}_1(t), \dots, w^{x}_{n+1}(t)$ sorted from largest to smallest with ties broken arbitrarily.

\begin{lemma} \label{lemma:sortedBoundChain}
	For any epoch $k$, it holds
	\begin{align*}
		\Phi(t_{e_k}) - \Phi(t_{e_{k+1}}) \geq \frac{1}{4} \sum_{i=1}^{n-1} \left(w_{\pi_i}(t)-w_{\pi_{i+1}}(t)\right)^2.
	\end{align*}
\end{lemma}

\begin{proof}
	Analogous to the proof of \Cref{lemma:sortedBound}. \qed
\end{proof}

\begin{lemma} \label{lemma:potentialRatioChain}
	Suppose that $\Phi(t_{e_k}) > 0$.
	Then,

	\begin{align*}
		\frac{\Phi(t_{e_k}) -\Phi(t_{e_{k+1}})}{\Phi(t_{e_k})} \geq \frac{1}{4 (n+1)^2}.
	\end{align*}
\end{lemma}
\begin{proof}
	\Cref{lemma:sortedBound} leads to

	\begin{align*}
		\frac{\Phi(t_{e_k}) -\Phi(t_{e_{k+1}})}{\Phi(t_{e_k})} &\geq \frac{1}{4} \frac{ \sum_{\pi_\ell=1}^{n} \left(w_{\pi_\ell}(t_{e_k}) -w_{\pi_{\ell+1}}(t_{e_k})\right)^2}{ \sum_{\pi_\ell=1}^{n+1} \left(w_{\pi_\ell}(t_{e_k}) -\overline{x}\right)^2.} \\
	\end{align*}

	Observe that the right-hand side does not change if we multiply each $w_{\pi_i}(t_{e_k})$ with the same constant.
	Additionally, it also does not change if we add the same constant to each $w_{\pi_i}(t_{e_k})$.
	Hence, we can assume w.l.o.g.\  $\sum_{\pi_\ell=1}^{n+1} w_{\pi_\ell}(t_{e_k}) = 0$ and $\sum_{\pi_\ell=1}^{n+1} \left(w_{\pi_\ell}(t_{e_k})\right)^2 = 1$ and obtain

	\begin{align*}
		\frac{\Phi(t_{e_k}) -\Phi(t_{e_{k+1}})}{\Phi(t_{e_k})} &\geq \quarter \min_{\substack{w_1 \geq w_2, \dots, \geq w_{n+1}  \\ \sum_{i}w_i = 0 \\\sum_{i}(w_i-w_1)^2 = 1}} \sum_{i=1}^{n} \left(w_i-w_{i+1}\right)^2
	\end{align*}

	The assumption $\sum_{i}w_i^2 =1$ implies that the average value of all $w_i^2$ is $\frac{1}{n}$ and hence there is at least some $j$ with $|w_j| \geq \frac{1}{\sqrt{n}}$.
	W.l.o.g. assume this $w_j$ is positive.
	The case $w_J < 0$ can be handled by multiplying each $w_i$ with $-1$ and sorting the elements in descending order.

	Now define $u_i = w_i - w_{i+1}$ for $i < n+1$ and $u_{n+1} = 0$.
	It holds $u_i \geq 0$ for all $i$ and $\sum_{i=1}^{n} u_i = w_1 - w_{n+1}$.
	Since  $w_j \geq \frac{1}{\sqrt{n}}$ and $\sum_{i}w_i = 0$, we can conclude $u_{n+1} < 0$ and thus $\sum_{i}u_i \geq \frac{1}{\sqrt{n}}$.

	As a final step, we obtain

	\begin{align*}
		\frac{\Phi(t_{e_k}) -\Phi(t_{e_{k+1}})}{\Phi(t_{e_k})} &\geq \quarter \min_{\substack{u_i  \geq 0 , \sum_{i}u_i \geq 1/(2 \sqrt{n})}} \sum_{i=1}^{n+1} u_i^2.
	\end{align*}

	The solution of the minimization problem is $u_i = \frac{1}{(n+1)^{3/2}}$ for each $i$.

	Hence,
	\begin{align*}
		\frac{\Phi(t_{e_k}) -\Phi(t_{e_{k+1}})}{\Phi(t_{e_k})} &\geq  \quarter \cdot \frac{1}{(n+1)^2}
	\end{align*}
	\qed
\end{proof}

\begin{proof}[Proof of \Cref{theorem:chainFormationTheorem}]
	Fix any epoch $e_k$.
	By \Cref{lemma:potentialRatioChain}, we obtain $\Phi(t_{e_{k+1}}) \leq (1-\frac{1}{4n^2}) \cdot \Phi(t_{e_k})$ and thus $\Phi(t_{e_{k+x}}) \leq \left(1-\frac{1}{4n^2}\right)^x \cdot \Phi(t_{e_k})$.
	Observe that $(1-y)^x \leq e^{-y\cdot x}$ where $e$ denotes Euler's number.
	Thus, choosing $x \geq 4(n+1)^2 \cdot \ln(\frac{1}{r})$ yields $\Phi(t_{e_{k+x}}) \leq r\cdot \Phi(t_{e_k})$.
	Since $\Phi(t_{e_k}) < n+1$, $r \leq \frac{\varepsilon}{n+1}$ leads to  $\Phi(t_{e_{k+x}}) \leq \varepsilon$ and the theorem follows.
	\qed
\end{proof}

%%%%%%%%%%%%%%%%%%%%%%%%%%%%%%%%%%%%%%%%%%%%%%%%%%%%%%%%%%%%%%%%%%%%%%%%%%%%%%%%%%%%%%%%%%%%%%%%%%%
%%%%%%%%%%%%%%%%%%%%%%%%%%%%%%%%%%%%%%%%%%%%%%%%%%%%%%%%%%%%%%%%%%%%%%%%%%%%%%%%%%%%%%%%%%%%%%%%%%%

%%
%% Bibliography
%%

%% Either use bibtex (recommended),

%\bibliographystyle{plain}
%\bibliography{threeDim}
%
%%% .. or use the thebibliography environment explicitly
%
%%%%%%%%%%%%%%%%%%%%%%%%%%%%%%%%%%%%%%%%%%%%%%%%%%%%%%%%%%%%%%%%%%%%%%%%%%%%%%%%%%%%%%%%%%%%%%%%%%%%
%%%%%%%%%%%%%%%%%%%%%%%%%%%%%%%%%%%%%%%%%%%%%%%%%%%%%%%%%%%%%%%%%%%%%%%%%%%%%%%%%%%%%%%%%%%%%%%%%%%%
%\newpage
%\appendix

\end{document}